%% file: main.tex
\documentclass[final,runningheads,envcountsect,envcountsame]{llncs}
\usepackage{ifthen}
\newboolean{arxivversion}

\setboolean{arxivversion}{true}

\usepackage{bm}
\usepackage{textcomp}
\usepackage{pgfplots}
\pgfplotsset{width=10cm,compat=1.4}
\usepackage{algorithm}
\usepackage{ifdraft}
\usepackage[noend]{algpseudocode}
\usepackage{mathtools}
\usepackage{enumerate}
\usepackage{cite}
\usepackage{nicefrac}
\usepackage{subcaption}
\usepackage{microtype}
\usepackage{xcolor,colortbl}
\usepackage{multirow}
\usepackage{amssymb}
\usepackage{forest}
\usepackage{wrapfig}
\usepackage{ stmaryrd }
\usepackage{paralist}
\forestset{qtree/.style={for tree={parent anchor=south, 
           child anchor=north,align=center,inner sep=0pt}}}
\usepackage{flexisym}
\usepackage{framed}
\usepackage{enumitem}
\usepackage{float}
\usepackage[most]{tcolorbox}
\usepackage{bbding}

\usepackage{booktabs}
\usepackage{adjustbox}
\usepackage{pgf}
\newcommand\Tstrut{\rule{0pt}{2.4ex}} 

\makeatletter
\newcommand{\otherlabel}[2]{\protected@edef\@currentlabel{#2}\label{#1}}
\makeatother

\usepackage{etoolbox}
\usepackage[colorinlistoftodos,prependcaption,textsize=tiny,disable]{todonotes}
\usepackage{xargs}                      
\newcommandx{\unsure}[2][1=]{\todo[linecolor=red,backgroundcolor=red!25,bordercolor=red,#1]{#2}}
\newcommandx{\change}[2][1=]{\todo[linecolor=blue,backgroundcolor=blue!25,bordercolor=blue,#1]{#2}}
\newcommandx{\toadd}[2][1=]{\todo[linecolor=pink,backgroundcolor=pink!25,bordercolor=pink,#1]{#2}}

\newcommandx{\sj}[2][1=]{\todo[linecolor=orange,backgroundcolor=orange!25,bordercolor=orange,#1]{SJ:#2}}
\newcommandx{\lk}[2][1=]{\todo[linecolor=green,backgroundcolor=green!25,bordercolor=green,#1]{LK:#2}}

\newcommandx{\jr}[2][1=]{\todo[linecolor=blue,backgroundcolor=blue!25,bordercolor=blue,#1]{JR:#2}}

\makeatletter
\RequirePackage[bookmarks,unicode,colorlinks=true]{hyperref}%
\def\@citecolor{blue}%
\def\@urlcolor{blue}%
\def\@linkcolor{blue}%
\def\orcidID#1{\smash{\protect\raisebox{-1.25pt}{\protect\href{http://orcid.org/#1}{\includegraphics{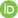}}}}}
\makeatother

\counterwithin{example}{section}
\counterwithin{definition}{section}
\usepackage{cleveref}

\usepackage{tikz}

\usetikzlibrary{arrows.meta,
                chains,
                positioning,
                quotes,
                shapes.geometric}

\tikzset{
    recgray/.style={draw,minimum width=2cm,minimum height=2cm,align=center,text width=2cm,fill=green!15!white,font=\sffamily},
    recgray2/.style={draw,minimum width=2cm,minimum height=1.3cm,align=center,text width=2cm,fill=green!15!white,font=\sffamily},
    teacher/.style={draw,minimum width=2cm,minimum height=3.2cm,align=center,text width=10cm,fill=yellow!10!white,font=\sffamily},
    sul/.style={draw,minimum width=2cm,minimum height=4.5cm,align=center,text width=2cm,fill=teal!30!white,font=\sffamily},
    learner/.style={draw,minimum width=2cm,minimum height=4.5cm,align=center,text width=2cm,fill=teal!30!white,font=\sffamily},
}

\usetikzlibrary{automata, positioning, arrows.meta, fit}
\tikzset{
>=stealth, 
node distance=3cm, 
every state/.style={thick, fill=gray!10}, 
initial text=$ $, 
}
\usetikzlibrary{patterns, intersections, positioning, shapes.geometric, calc,
	automata, arrows, decorations.pathreplacing,decorations.pathmorphing, patterns.meta,}

\makeatletter
\tikzset{
        hatch distance/.store in=\hatchdistance,
        hatch distance=5pt,
        hatch shift/.store in=\hatchshift,
        hatch shift=0pt,
        }
\pgfdeclarepatternformonly[\hatchdistance,\hatchshift]{north east hatch}
    {\pgfqpoint{-1pt}{-1pt}}
    {\pgfqpoint{\hatchdistance}{\hatchdistance}}
    {\pgfpoint{\hatchdistance-1pt}{\hatchdistance-1pt}}%
    {
        \pgfsetcolor{\tikz@pattern@color}
        \ifdim \hatchshift > 0pt
          \pgfpathmoveto{\pgfqpoint{0pt}{\hatchshift}}
          \pgfpathlineto{\pgfqpoint{\dimexpr0.15pt+\hatchshift}{-0.15pt}}
        \fi
        \pgfpathmoveto{\pgfqpoint{0pt}{0pt}}
        \pgfpathlineto{\pgfqpoint{\hatchdistance}{\hatchdistance}}
        \pgfusepath{stroke}
    }
\makeatother
\usetikzlibrary{cd}
\tikzset{
	treenode/.style = {align=center, inner sep=0pt, text centered},
  basis/.style = {
    pattern=north east lines,
    pattern color=magenta!70!black,
  },
  newbasis/.style={
    pattern=north east lines,
    pattern color=teal!80!black!60!white,
  },
  other/.style={
    pattern=north east lines,
    pattern color=red,
  },
  frontier/.style={
    pattern=north east lines,
    pattern color=yellow!80!black,
  },
  pw/.style={
    pattern=north west lines,
    pattern color=orange!70!black!70!white,
  },
  q0class/.style={
    pattern=north west lines,
    pattern color=red!40!white,
  },
  q1class/.style={
    pattern=north west lines,
    pattern color=blue!60!white,
  },
  q2class/.style={
    pattern=crosshatch,
    pattern color=green!50!black!60!white,
  },
  basic/.style = {
    fill=white,
  },
  sharedbasis/.style = {
    pattern=north west lines,
    pattern color=orange!50!black!60!white,
    distance=2pt,
    preaction={%
      pattern=north east lines,
      pattern color=teal!80!black!60!white,
      distance=2pt,
      yshift=1.5pt,
    }
  },
}
\usepackage[outline]{contour}
\contourlength{1.2pt}

\usepackage{svg}

\pgfplotsset{compat=1.8}

\pgfplotsset{vasymptote/.style={
    before end axis/.append code={
        \draw[densely dashed] ({rel axis cs:0,0} -| {axis cs:#1,0})
        -- ({rel axis cs:0,1} -| {axis cs:#1,0});
    }
}}

\definecolor{darkred}{rgb}{.75,0,0}

\definecolor{darkgreen}{rgb}{0,.75,0}

\newcommand{\access}{\ensuremath{\mathsf{access}}}
\newcommand{\accessT}{\ensuremath{\mathsf{access}^{\Obs}\hspace*{-0.05cm}}}
\newcommand{\partialto}{\ensuremath{\rightharpoonup}}

\newcommand{\totalapart}{\ensuremath{\apart_{\uparrow}}}
\newcommand{\apart}{\ensuremath{\mathrel{\#}}}

\newcommand{\takeout}[1]{\relax}

\newcommand{\aaal}[0]{adaptive AAL}
\newcommand{\lsharp}[0]{$L^{\#}$}

\newcommand{\basis}[0]{B}
\newcommand{\frontier}[0]{F}
\newcommand{\firstbasisstate}[0]{\ensuremath{q}}
\newcommand{\secondbasisstate}[0]{\ensuremath{q'}}
\newcommand{\thirdbasisstate}[0]{\ensuremath{q''}}
\newcommand{\firstfrontierstate}[0]{\ensuremath{r}}

\newcommand{\firstrefbasisstate}[0]{\ensuremath{p}}
\newcommand{\secondrefbasisstate}[0]{\ensuremath{p'}}

\newcommand{\nrrefstates}[0]{\ensuremath{o}}

\renewcommand{\H}{\mathcal{H}}
\renewcommand{\S}{\mathcal{S}}
\newcommand{\R}{\mathcal{R}}
\newcommand{\X}{\mathcal{X}}
\newcommand{\M}{\mathcal{M}}

\newcommand{\Obs}{\mathcal{T}}
\newcommand{\Hyp}{\mathcal{H}}
\newcommand{\bigO}{\mathcal{O}}

\newcommand{\converges}{\ensuremath{\mathord{\downarrow}}}
\newcommand{\diverges}{\ensuremath{\mathord{\uparrow}}}

\newcommand{\code}[1]{\texttt{#1}}

\usepackage{xspace}
\newcommand{\defineApiFunction}[1]{%
\expandafter\newcommand\csname #1\endcsname{\text{\upshape\textsc{#1}}\xspace}%
}

\newcommand{\lsharprebuild}{$L^{\#}_{\scalebox{0.7}{R}}$}
\newcommand{\lsharpmatch}{$L^{\#}_{\scalebox{0.7}{\matches}}$}
\newcommand{\lsharprebuildmatch}{$L^{\#}_{\scalebox{0.7}{R},\scalebox{0.7}{\matches}}$}
\newcommand{\lsharpapproxmatch}{$L^{\#}_{\scalebox{0.7}{\approxmatches}}$}

\renewenvironment{proof}{\textit{Proof.}}{\hfill \ensuremath{\square}}
  
\newcommand{\myparagraph}[1]{\smallskip\noindent \emph{#1}}
\newcommand{\mysubsubsection}[1]{\medskip\noindent \textbf{#1}}

\defineApiFunction{LSharp}
\defineApiFunction{GetAccessSeq}
\defineApiFunction{GetExpert}
\defineApiFunction{ComputeSubalphabet}
\defineApiFunction{GetInfixSeq}
\defineApiFunction{GetSeparatingSeq}

\defineApiFunction{RebuildingProc}
\defineApiFunction{CheckConsistency}
\defineApiFunction{BuildHypothesis}
\defineApiFunction{EquivalenceQuery}
\defineApiFunction{ProcCounterEx}
\defineApiFunction{OutputQuery}
\defineApiFunction{Ads}
\defineApiFunction{TestHypothesisMAB}
\defineApiFunction{GenerateTestCaseWithExpert}
\defineApiFunction{UpdateProbabilities}
\defineApiFunction{UpdateWeights}
\defineApiFunction{Learn}
\defineApiFunction{MAT}
\defineApiFunction{MABEQ}

\newcommand{\pdlstar}{$\partial L^*_M$}
\newcommand{\IKV}{\textit{IKV}}
\newcommand{\adaptivelsharp}{$AL^{\#}$\xspace}
\newcommand{\matches}{\scalebox{0.65}{\ensuremath{{}\overset{\surd}{=}{}}}}
\newcommand{\approxmatches}{\scalebox{0.65}{\ensuremath{{}\overset{\surd}{\simeq}{}}}}
\newcommand{\notmatches}{\scalebox{0.65}{\ensuremath{{}\overset{\surd}{\neq}{}}}}
\newcommand{\notapproxmatches}{\scalebox{0.65}{\ensuremath{{}\overset{\surd}{\not\simeq}{}}}}

\newenvironment{ruleenv}[2]{\smallskip\noindent\textbf{Rule (#1): #2}.}{}

\algdef{SE}[MYWHILE]{MyWhile}{MyWhile}{\algorithmicwhile\#1}%

\definecolor{ballblue}{rgb}{0.13, 0.67, 0.8}

\algdef{SE}[SUBALG]{Indent}{EndIndent}{}{\algorithmicend\ }%
\algtext*{Indent}
\algtext*{EndIndent}

\tikzset{
  do guard/.style={
    inner xsep=0pt,
  },
  guard line offset/.style={
    xshift=2mm,
  },
  bigtalloblong/.style={
    draw=black,
    minimum width=3pt,
    minimum height=1em,
    inner xsep=0pt,
  },
}
\newcommand{\connectDoGuards}[2]{%
  \draw[overlay,draw=black!40!white] ([guard line offset]#1) -- ([guard line offset]#2);
}

\algtext*{EndWhile}%

\algblockdefx{DoIf}{EndDoIf}
[1]{\begin{tikzpicture}[remember picture,baseline=(do.base)]
    \node[do guard] (do) {\textbf{do}};
    \coordinate (last do) at (do.south west);
  \end{tikzpicture}
     #1 $\rightarrow$}
   {\begin{tikzpicture}[remember picture,baseline=(od.base)]
       \node[do guard] (od) {X};
       \connectDoGuards{opening do}{od.base west}
    \end{tikzpicture}\textbf{od}}
\algcblockdefx[DoIf]{DoIf}{ElsDoIf}{EndDoIf}
[1]{\begin{tikzpicture}[remember picture,baseline=(new do.base)]
    \node[do guard] (new do) {\phantom{\bfseries d}};
    \draw[draw=none] ([yshift=3pt]new do.north);
    \connectDoGuards{last do}{new do.north west}
    \coordinate (last do) at (new do.south west);
    \end{tikzpicture}
  #1 $\rightarrow$}
{\begin{tikzpicture}[remember picture,baseline=(od.base)]
       \node[do guard] (od) {\textbf{end do}};
       \connectDoGuards{last do}{od.north west}
    \end{tikzpicture}%
  }

\newcommand{\StateIf}[1]{\State\textbf{if}~#1~\textbf{then:}}
\newcommand{\StateElse}{\State\textbf{else:}~}
\algloopdefx{IfLoop}[1]{\textbf{If} #1 \textbf{then}}

\begin{document}

\title{State Matching and Multiple References in Adaptive Active Automata Learning\thanks{%
This research is partially supported by the NWO grant No.~VI.Vidi.223.096.}}

\titlerunning{ State Matching and Multiple References in \aaal}
%
\author{Loes Kruger\Envelope \orcidID{0009-0003-3275-6806} 
\and  Sebastian Junges \orcidID{0000-0003-0978-8466} 
\and Jurriaan Rot \orcidID{0000-0002-1404-6232}
  }
\authorrunning{L. Kruger et al.}
%
\institute{
 Radboud University, Nijmegen, the Netherlands\\
 \email{\{loes.kruger,sebastian.junges,jurriaan.rot\}@ru.nl}
}
\maketitle              

\begin{abstract}
	Active automata learning (AAL) is a method to infer state machines by interacting with black-box systems.
	Adaptive AAL aims to reduce the sample complexity of AAL by incorporating domain specific knowledge in the form of  (similar) reference models.
	Such reference models appear naturally when learning multiple versions or variants of a software system. 
	In this paper, we present state matching, which allows flexible use of the structure of these reference models by the learner. 
	State matching is the main ingredient of adaptive \lsharp, a novel framework for adaptive learning, built on top of \lsharp{}. 
    Our empirical evaluation shows that adaptive \lsharp{} improves the state of the art by up to two orders of magnitude. 
    
    \end{abstract}

\input{introduction.tex}


\input{overview.tex}

\input{preliminaries.tex}
\input{lsharprebuilding.tex}

\input{statematching.tex}
\input{adaptivelsharp.tex}

\input{multiplereferences.tex}
\input{experiments.tex}
\input{conclusion.tex}
 
\newpage
\bibliographystyle{plainurl}
\bibliography{references.bib}

\ifthenelse{\boolean{arxivversion}}{
\newpage
\appendix
\input{appendix.tex}

\input{appendixexample.tex}}

\end{document}

%% file: introduction.tex
\section{Introduction} \label{sec:introduction}
\emph{Automata learning} aims to extract state machines from observed input-output sequences of some system-under-learning (SUL). 
\emph{Active} automata learning (AAL) assumes that one has black-box access to this SUL, allowing the learner to incrementally choose inputs and observe the outputs. The models learned by AAL can be used as a documentation effort, but are more typically used as basis for testing, verification, conformance checking, fingerprinting---see~\cite{DBLP:journals/cacm/Vaandrager17,DBLP:conf/dagstuhl/HowarS16} for an overview of applications.
The classical algorithm for AAL is $L^*$, introduced by Angluin~\cite{DBLP:journals/iandc/Angluin87};
state-of-the-art algorithms are, e.g., \lsharp~\cite{DBLP:conf/tacas/VaandragerGRW22} and TTT~\cite{DBLP:conf/rv/IsbernerHS14}, which are available in toolboxes such as LearnLib~\cite{DBLP:conf/cav/IsbernerHS15} and AALpy~\cite{DBLP:conf/atva/MuskardinAPPT21}.

The primary challenge in AAL is to reduce the number of inputs sent to the SUL, referred to as the \emph{sample complexity}. To learn a 31-state machine with 22 inputs, state-of-the-art learners may send several million inputs to the SUL~\cite{DBLP:conf/tacas/VaandragerGRW22}. 
This is not necessarily unexpected: the underlying space of 31-state state machines is huge and it is nontrivial how to maximise information gain.
The literature has investigated several approaches to accelerate learners, see the overview of~\cite{DBLP:journals/cacm/Vaandrager17}. Nevertheless, scalability remains a core challenge for AAL.


We study \emph{adaptive} AAL~\cite{DBLP:journals/igpl/GrocePY06}, which aims to improve the sample efficiency by utilizing expert knowledge already given to the learner. In (regular) AAL, a learner commonly starts learning from scratch. In adaptive AAL, however, the learner is given a \emph{reference model}, which ought to be similar to the SUL. Reference models occur naturally in many applications of AAL. For instance:\sj{shortened these a bit}
(1)~Systems evolve over time due to, e.g., bug fixes or new functionalities---and we may have learned the previous system;
(2)~Standard protocols may be implemented by a variety of tools;
(3)~The SUL may be a variant of other systems, e.g., being the same system executing in another environment, or a system configured differently. 

Several algorithms for adaptive AAL have been proposed \cite{DBLP:journals/igpl/GrocePY06, DBLP:conf/cbse/WindmullerNSHB13, DBLP:journals/fmsd/ChakiCSS08, DBLP:conf/ifm/DamascenoMS19, DBLP:conf/birthday/FerreiraHS22}. 
Intuitively, the idea is that these methods try to \emph{rebuild} the part of the SUL which is similar to the reference model. This is achieved by deriving suitable queries from the reference model, using so-called \emph{access sequences} to reach states, and so-called \emph{separating sequences} to distinguish these from other states.\lk{I would prefer to remove so-called and put emph around access and sep seq }\sj{ok with the emph, but i would keep the so-called as we have not introduced them and do not explain them here? It indicates that a reader does not need the definition}
These algorithms rely on a rather strict notion of similarity that depends on the way we reach these states\sj{pls check}. In particular, existing rebuilding algorithms cannot effectively learn an SUL from a reference model that has a different initial state, see Sec.~\ref{sec:overview}. \lk{Should be merged more with the related work possibly. Also feels a bit incomplete }

We propose an approach to adaptive AAL based on \emph{state matching}, which allows flexibly identifying parts of the unknown SUL where the reference model may be an informative guide\sj{removed: to search for queries}. 
 More specifically, in this approach, we match states in the model that we have learned so far (captured as a tree-shaped automaton) with states in the reference model such that the outputs agree on all enabled input sequences.
  This matching allows for targeted re-use of separating sequences from the reference model and is independent of the access sequences. We refine the approach by using \emph{approximate state matching}, where we match a current state with one from the reference model that agrees on most inputs.

Approximate state matching is the essential ingredient for the novel \adaptivelsharp algorithm. This algorithm is a conservative extension of the recent \lsharp{}~\cite{DBLP:conf/tacas/VaandragerGRW22}. Along with approximate state matching, \adaptivelsharp includes rebuilding steps, which are similar to existing methods, but tightly integrated in \lsharp. 
Finally, \adaptivelsharp is the first approach with \emph{dedicated support} to use more than one reference model. 

\myparagraph{Contributions.}
We make the following contributions to the state-of-the-art in adaptive AAL. 
First, we present state matching and its generalization to approximate state matching which allows flexible re-use of separating sequences from the reference model.
Second, we include state matching and rebuilding in an unifying approach, called \adaptivelsharp, which generalizes the \lsharp{} algorithm for non-adaptive automata learning. 
We analyse the resulting framework in terms of termination and complexity.
This framework naturally supports using multiple reference models as well as removing and adding inputs to the alphabet.
Our empirical results show the efficacy of \adaptivelsharp. In particular, \adaptivelsharp may reduce the number of inputs to the SUL by two orders of magnitude. \lk{Was in rel work not in intro: we consider pdlstar and IKV state of the art. Maybe put in experiments?}

\myparagraph{Related work.} 
Adaptive AAL goes back to~\cite{DBLP:journals/igpl/GrocePY06}. That paper, and many of the follow-up approaches~\cite{DBLP:journals/fmsd/ChakiCSS08,DBLP:conf/birthday/BainczykSH20,DBLP:conf/ifm/DamascenoMS19,DBLP:conf/birthday/FerreiraHS22} re-use access sequences and separating sequences from the reference model (or from the data structures constructed when learning that model). 
The recent approach in~\cite{DBLP:conf/ifm/DamascenoMS19} removes redundant access sequences during rebuilding and continues learning with informative separating sequences.
In~\cite{DBLP:conf/cbse/WindmullerNSHB13}, an $L^*$-based adaptive AAL approach is proposed where the algorithm starts by including \emph{all} separating sequences that arise when learning the reference model with $L^*$, ignoring access sequences. This algorithm is used in~\cite{DBLP:conf/fmics/HuistraMP18} for a general study of the usefulness of adaptive AAL: Among others, the authors suggest using more advanced data structures than the observation tables in $L^*$. Indeed, in~\cite{DBLP:conf/birthday/BainczykSH20} the internal data structure of the TTT algorithm is used~\cite{DBLP:conf/rv/IsbernerHS14} in the context of lifelong learning; the precise rebuilding approach is not described\lk{Too defensive?}. The recent~\cite{DBLP:conf/birthday/FerreiraHS22} proposes an adaptive AAL method based on discrimination trees as used in the Kearns-Vazirani algorithm~\cite{DBLP:books/daglib/0041035}. We consider the algorithms proposed in~\cite{DBLP:conf/ifm/DamascenoMS19,DBLP:conf/birthday/FerreiraHS22} the state-of-the-art and have experimentally compared \adaptivelsharp in Sec.~\ref{sec:experiments}.

%% file: overview.tex
\section{Overview} \label{sec:overview}
\input{figures/runningexample.tex}

We illustrate (1)~how adaptive AAL uses a reference model to help learn a system and (2)~how this may reduce the sample complexity of the learner.

\myparagraph{MAT framework.} We recall the standard setting for AAL: Angluin's MAT framework, cf.\ \cite{DBLP:journals/cacm/Vaandrager17,DBLP:conf/dagstuhl/HowarS16}. Here, the learner has no direct access to the SUL, but may ask \emph{output queries (OQs)}: these return, for a given input sequence, the sequence of outputs from the SUL; and \emph{equivalence queries (EQs)}: these take a Mealy machine $\H$ as input, and return whether or not $\H$ is equivalent to the SUL. In case it is not, a counterexample is provided in the form of a sequence of inputs for which $\H$ and the SUL return different outputs. EQs are expensive~\cite{DBLP:conf/wcre/YangASLHCS19,DBLP:journals/sosym/AslamCSB20,DBLP:conf/uss/RuiterP15,DBLP:journals/corr/abs-1904-07075}, therefore, we aim to learn the SUL using primarily OQs. 

\myparagraph{Apartness.} Learning algorithms in the MAT framework typically assume that two states are equivalent as long as their \emph{known} residual languages are equivalent. To \emph{discover a new state}, we must therefore (1)~access it by an input sequence and (2)~prove this state distinct (\emph{apart}) from the other states that we already know.
Consider the SUL $\S$ in Fig.~\ref{fig:target}. The access sequences $c$, $ca$ access $q_4$ and $q_5$, respectively, from the initial state. These states are different because the response to executing $c$ from $q_4$ and $q_5$ is distinct: We say $c$ is a \emph{separating sequence} for $q_4$ and $q_5$.
 This difference can be observed by posing OQs for $cc$ and $cac$, consisting of the access sequences for $q_4$ and $q_5$ followed by their separating sequence $c$.

\myparagraph{Aim.} The aim of adaptive AAL is to learn SULs with fewer inputs, using knowledge in the form of a reference model, known to the learner and preferrably similar to the SUL. The discovery of states is accelerated by extracting candidates for both (1)~access sequences and (2)~separating sequences from the reference model.

\myparagraph{Rebuilding.}
The state-of-the-art in adaptive AAL uses access sequences and separating sequences from the reference model~\cite{DBLP:conf/ifm/DamascenoMS19,DBLP:conf/birthday/FerreiraHS22} in an initial phase.
 Consider the Mealy machine $\R_1$ in Fig.~\ref{fig:ref1} as a reference model for the SUL $\S$ in Fig.~\ref{fig:target}.
 The sequences $\varepsilon$, $c$, $ca$ can be used to access all orange states in both $\S$ and $\R_1$.
 The separating sequences $c$ and $ac$ for these states in $\R_1$ also separate the orange states in $\S$. By asking OQs combining the access sequences and separating sequences, we discover all orange states for $\S$. 
 
 \myparagraph{Limits of rebuilding.}
 However, these rebuilding approaches have limitations. 
 Consider $\R_2$ in Fig.~\ref{fig:ref2}. The sequences $\varepsilon$, $b$, $bb$ and $bbb$ can be used to access all states in $\R_2$. Concatenating these with any separating sequences from $\R_2$ will not be helpful to learn SUL $\S$, because in $\S$ these sequences all access $q_0$.
 However, the separating sequences from $\R_2$ are useful if executed in the right state of $\S$. For instance, the sequence $bb$ separates all states in $\R_2$, and the blue states in $\S$.
Thus, rebuilding does not realise the potential of reusing the separating sequences from $\R_2$, since the access sequences for the relevant states are different.

\myparagraph{State Matching.} We extend adaptive AAL with \emph{state matching}. State matching overcomes the strong dependency on the access sequences and allows the efficient usage of reference models where the residual languages of the individual states are similar. Suppose that while learning, we have not yet separated $q_0$ and $q_1$ in $\S$, but we do know the output of the $b$-transition from $q_0$. We may use that output to \emph{match} $q_0$ with $p_3$ in $\R_2$: these two states agree on input sequences where both are defined. Subsequently, we can use the separating sequence $bb$ between $p_3$ and $p_0$ to separate $q_0$ and $q_1$, through OQs $bb$ and $abb$. 

\myparagraph{Approximate State Matching.} It rarely happens that states in the SUL exactly match states in the reference model: Consider the scenario where we want to learn $\S$ with reference model $\R_3$ from Fig.~\ref{fig:ref3}. States $q_0$ and $s_3$ do not match because they have different outputs for input $b$ but are still similar. This motivates an approximate version of matching, where a state is matched to the reference state which maximises the number of inputs with the same output.

\myparagraph{Outline.} After the preliminaries (Sec.~\ref{sec:prelim}), we recall the \lsharp{} algorithm and extend it with rebuilding (Sec.~\ref{sec:lsharprebuild}). We then introduce 
 adaptive AAL with state matching and its approximate variant (Sec.~\ref{sec:statematch}). Together with rebuilding, this results in the \adaptivelsharp algorithm
(Sec.~\ref{sec:alsharp}). We proceed to define a variant that allows the use of multiple reference models (Sec.~\ref{sec:multiple}). 
This is helpful already in the example discussed in this section: given both $\R_1$ and $\R_2$, \adaptivelsharp with multiple reference models allows to discover all states in $\S$ without any EQs, see \ifthenelse{\boolean{arxivversion}}{App.~\ref{app:F}.}{App.~F~of~\cite{Kruger2024AdaptiveArxiv}.} 

%% file: figures/runningexample.tex
\begin{figure}[t]
    \resizebox{.99\textwidth}{!}{
     \centering
    \begin{subfigure}[b]{0.32\textwidth}
        \scalebox{0.57}{
          \includegraphics{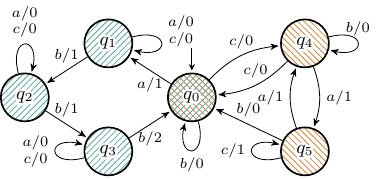}
        }
        \caption{$\S$}
        \label{fig:target}
    \end{subfigure}
    \begin{subfigure}[b]{0.17\textwidth}
    \scalebox{0.57}{
        

    
        \includegraphics{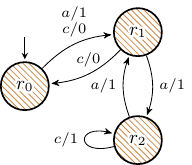}}
    \caption{$\R_1$}
    \label{fig:ref1}
     \end{subfigure}
    \begin{subfigure}[b]{0.2\textwidth}
        \scalebox{0.57}{
            

        
            \includegraphics{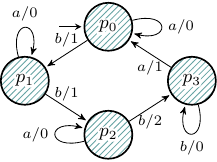}}
        \caption{$\R_2$}
        \label{fig:ref2}
        \end{subfigure}
    
    \begin{subfigure}[b]{0.2\textwidth}
    \scalebox{0.57}{
     \includegraphics{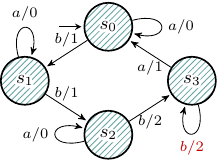}}	
     \caption{$\R_3$}
     \label{fig:ref3}
    \end{subfigure}
}
    \caption{An SUL $\S$ and three reference models $\R_1$, $\R_2$ and $\R_3$.}
    \label{fig:runningexample}
\end{figure}

%% file: preliminaries.tex
\section{Preliminaries} \label{sec:prelim}
For a partial map $f \colon X \partialto Y$, we write $f(x)\converges$ if $f(x)$ is defined and $f(x)\diverges$ otherwise.

\begin{definition}
A \emph{partial Mealy machine} is a tuple $\M = (Q, I, O, q_0, \delta, \lambda)$, where $Q$, $I$ and $O$ are finite sets of states, inputs and outputs respectively; $q_0 \in Q$ an \emph{initial state}, $\delta\colon Q \times I \partialto Q$ a \emph{transition function}, and $\lambda\colon Q \times I \partialto O$ an \emph{output function} such that $\delta$ and $\lambda$ have the same domain.
 A \emph{(complete) Mealy machine} is a partial Mealy machine where $\delta$ and $\lambda$ are total. If not specified otherwise, a Mealy machine is assumed to be complete.
 \end{definition}
We write $\M|_{I}$ to denote $\M$ restricted to alphabet $I$.
We use the superscript $\M$ to indicate to which Mealy machine we refer, e.g. $Q^{\M}$ and $\delta^{\M}$. 
The transition and output functions are naturally extended to input sequences of length $n \in \mathbb{N}$ as functions $\delta\colon Q \times I^n \partialto Q$ and $\lambda\colon Q \times I^n \partialto O^n$. 
We abbreviate $\delta(q_0, w)$ by $\delta(w)$. 

\begin{definition}\label{def:matching}
   Let $\M_1$, $\M_2$ be partial Mealy machines. States $p \in Q^{\M_1}$ and $q \in Q^{\M_2}$ \emph{match}, written $p \matches q$, if $\lambda(p,\sigma)=\lambda(q,\sigma)$ for all $\sigma \in (I^{\M_1} \cap I^{\M_2})^*$ with $\delta(p,\sigma){\converges}$ and $\delta(q,\sigma){\converges}$. 
    If $p$ and $q$ do not match, they are \emph{apart}, written $p \apart q$. 
\end{definition}
If $p \apart q$, then there is a \emph{separating sequence}, i.e., a sequence $\sigma$ such that $\lambda(p,\sigma) \neq \lambda(q,\sigma)$; this situation is denoted by $\sigma \vdash p \apart q$.
The definition of matching allows the input (and output) alphabets of the underlying Mealy machines to differ; it requires that they agree on all commonly defined input sequences. If $\M_1$ and $\M_2$ are complete and have the same alphabet, then the matching of states is referred to as \emph{language equivalence}. Two complete Mealy machines are \emph{equivalent} if their initial states are language equivalent. 

Let $\M$ be a partial Mealy machine. A state $q \in Q^{\M}$ is \emph{reachable} if there exists $\sigma \in I^*$ such that $\delta^{\M}(q_0, \sigma)=q$. The \emph{reachable part} of $\M$ contains all reachable states in $Q^{\M}$.
A sequence $\sigma$ is an \emph{access sequence} for $q \in Q^{\M}$ if $\delta^{\M}(\sigma)=q$. 
A set $P \subseteq I^*$ is a \emph{state cover} for $\M$ if $P$ contains an access sequence for every reachable state in $\M$. 
In this paper, a \emph{tree} $\Obs$ is a partial Mealy machine where every state $q$ has a \emph{unique} access sequence, denoted by $\access(q)$. 

\begin{definition} \label{def:sepfam}
    Let $\M$ be a complete Mealy machine. A set $W_q \subseteq (I^{\M})^*$ is a \emph{state identifier} for $q \in Q^\M$ if for 
	all $p \in Q^\M$ with $p \apart q$ there exists $\sigma \in W_q$ such that $\sigma \vdash p \apart q$. 
	A \emph{separating family} is a collection of state identifiers $\{W_p\}_{p \in Q^\M}$ such that for all $p,q \in Q^\M$ with $p \apart q$ there exists $\sigma \in W_p \cap W_q$ with $\sigma \vdash p \apart q$.
\end{definition} 
We use $P^{\M}$ and $\{W_q\}^{\M}$ to refer to a minimal state cover and a separating family for $\M$ respectively. State covers and separating families can be constructed for every Mealy machine, but are not necessarily unique. 

%% file: lsharprebuilding.tex
\section{$L^{\#}$ with Rebuilding} \label{sec:lsharprebuild}
We first recall the \lsharp\ algorithm for (standard) AAL~\cite{DBLP:conf/tacas/VaandragerGRW22}. Then, we consider adaptive learning by presenting an \lsharp-compatible variant of rebuilding. 

\subsection{Observation Trees}
\lsharp~uses an observation tree as data structure to store the observed traces of $\M$.

\begin{definition}
    A tree $\Obs$ is an \emph{observation tree} if there exists a mapping $f \colon Q^\Obs \to Q^\M$ such that $f(q_0^{\Obs})=q_0^{\M}$ and $q \xrightarrow[]{i/o} q'$ implies $f(q) \xrightarrow[]{i/o} f(q')$.
\end{definition}

In an observation tree, a \emph{basis} is a subtree that describes unique behaviour present in the SUL. 
Initially, a basis $\basis \subseteq Q^{\Obs}$ contains the root state. All states in the basis are pairwise apart, i.e., for all $\firstbasisstate \neq \secondbasisstate \in \basis$ it holds that $\firstbasisstate \apart \secondbasisstate$.  
For a fixed basis, its \emph{frontier} is the set of states $\frontier \subseteq Q^{\Obs}$ which are immediate successors of basis states but which are not in the basis themselves.
\input{figures/M1Lexample.tex}
\begin{example}
    Fig.~\ref{fig:tree2} shows an observation tree $\Obs'$ for the Mealy machine $\H'$ from Fig.~\ref{fig:hyp2}. The separating sequences $c$ and $ac$ show that the states in basis $\basis = \{t_0, t_2, t_3 \}$ are all pairwise apart. The frontier $\frontier$ is $\{ t_1, t_4, t_5, t_6 \}$. 
\end{example}
We say that a frontier state is \emph{isolated} if it is apart from all basis states. A frontier state is \emph{identified} with a basis state $q$ if it is apart from all basis states except $q$. We say the observation tree is \emph{adequate} if all frontier states are identified, no frontier states are isolated and each basis state has a transition with every input.
If every frontier state is identified and each basis state has a transition for every input, the observation tree can be \emph{folded} to create a complete Mealy machine \ifthenelse{\boolean{arxivversion}}{(formalized in Def.~\ref{def:foldedback})}{}. The Mealy machine has the same states as the basis.
The transitions between basis states are the same as in the observation tree. Transitions from basis states to frontier states are \emph{folded back} to the basis state the frontier state is identified with. We call the resulting complete Mealy machine a \emph{hypothesis} whenever this canonical transformation is used.
\begin{example}
In $\Obs'$ (Fig.~\ref{fig:tree2}) the frontier states are identified as follows: $t_1\mapsto t_2, t_4 \mapsto t_3, t_5 \mapsto t_0$ and $t_6 \mapsto t_2$.
Hypothesis $\H'$ (Fig.~\ref{fig:hyp2}) can be folded back from $\Obs'$. The dashed transitions in Fig.~\ref{fig:hyp2} represent the folded transitions.
\end{example}

\subsection{The $L^{\#}$ Algorithm} \label{sec:lsharp}
The \lsharp~algorithm maintains an observation tree $\Obs$ and a basis $\basis$. Initially, $\Obs$ consists of just a root node $q_0$ and  $\basis = \{ q_0 \}$. We denote the frontier of $\basis$ by $\frontier$.
The \lsharp{} algorithm then repeatedly applies the following four rules. 

\begin{itemize}[noitemsep,topsep=0pt]
\item 
The \emph{promotion} rule (\textbf{P}) extends $B$ by $\firstfrontierstate \in F$ when $\firstfrontierstate$ is isolated.
\item The \emph{extension} rule (\textbf{Ex}) poses OQ $\access(\firstbasisstate)i$ for $\firstbasisstate \in \basis, i \in I$ with $\delta(\firstbasisstate, i)\diverges$. 
\item The \emph{separation} rule (\textbf{S}) takes a state $\firstfrontierstate \in \frontier$ that is not apart from $\firstbasisstate, \secondbasisstate \in \basis$ and poses OQ $\access(r)\sigma$ with $\sigma \vdash \firstbasisstate \apart \secondbasisstate$ that shows $\firstfrontierstate$ is apart from $\firstbasisstate$ or $\secondbasisstate$. 
\item The \emph{equivalence} rule (\textbf{Eq}) folds $\Obs$ into hypothesis $\H$, checks whether $\H$ and $\Obs$ agree on all sequences in $\Obs$ and poses an EQ. If $\H$ and the SUL are not equivalent, counterexample processing isolates a frontier state. 
\end{itemize}
The pre- and postconditions of the rules are summarized in (the top rows of) Table~\ref{tab:rules}. 
A detailed account is given in the paper introducing \lsharp~\cite{DBLP:conf/tacas/VaandragerGRW22}. 

\begin{table}[t]
    \caption{Extended \lsharp{} rules with parameters, preconditions and postconditions.}
    \resizebox{.99\textwidth}{!}{
    \begin{tabular}{p{0.22cm}|p{1.58cm}||p{1.9cm}|l|p{2cm}}
        & Rule & Parameters & Precondition & Postcondition \\ \hline \hline
        \parbox[t]{2mm}{\multirow{8}{*}{\rotatebox[origin=c]{90}{Sec.~\ref{sec:lsharp}}}} &
        \emph{promotion}\Tstrut & $\firstfrontierstate \in \frontier$ & $\forall \firstbasisstate \in \basis, \firstbasisstate \apart \firstfrontierstate$ & $\firstfrontierstate \in \basis$ \\ \cline{2-5}
        & \emph{extension}\Tstrut & $\firstbasisstate \in \basis$, $i \in I$ & $\delta^{\Obs}(\firstbasisstate,i){\uparrow}$ & $\delta^{\Obs}(\firstbasisstate,i){\converges}$ \\ 
        \cline{2-5}
        & \multirow{2}{=}{\emph{separation}}\Tstrut & $\firstfrontierstate \in \frontier,$ & $\neg(\firstfrontierstate \apart \firstbasisstate), \neg(\firstfrontierstate \apart \secondbasisstate), \firstbasisstate \neq \secondbasisstate$ & $\firstfrontierstate \apart \firstbasisstate \lor \firstfrontierstate \apart \secondbasisstate$ \\ 
        & & $\firstbasisstate, \secondbasisstate \in \basis$ & & \\ \cline{2-5}
        & \multirow{3}*{\emph{equivalence}}\Tstrut & - & $\forall \firstbasisstate \in \basis.~\forall i \in I.~\delta^{\Obs}(\firstbasisstate,i){\converges},$  & $\exists \firstfrontierstate \in \frontier$ s.t. \\ 
        &  & & $\forall \firstfrontierstate \in \frontier.~\exists \firstbasisstate \in \basis.$  & $\forall \firstbasisstate \in \basis.~\firstfrontierstate \apart \firstbasisstate$  \\ 
        & & & $(\neg(\firstfrontierstate \apart \firstbasisstate) \land ~\forall \secondbasisstate\in \basis\setminus\{\firstbasisstate\}.~\firstfrontierstate \apart \secondbasisstate)$ & \\ \hline
        \parbox[t]{2mm}{\multirow{7}{*}{\rotatebox[origin=c]{90}{Sec~\ref{sec:rebuild}}}} & \multirow{5}*{\emph{rebuilding}}\Tstrut & $\firstbasisstate,\secondbasisstate \in \basis,$ & $\delta^{\Obs}(\firstbasisstate,i) \notin \basis, \neg(\secondbasisstate \apart \delta^{\Obs}(\firstbasisstate,i))$, & $\delta^{\Obs}(\firstbasisstate,i\sigma){\converges}$, \\
        & & $i \in I$ & $\accessT(\firstbasisstate)i, \accessT(\secondbasisstate) \in P^{\R}$, & $\delta^{\Obs}(\secondbasisstate,\sigma){\converges}$ \\
        & & & $\sigma = \mathsf{sep}\bm{(}\delta^{\R}(\accessT(\firstbasisstate)i),\delta^{\R}(\accessT(\secondbasisstate))\bm{)}$, & {} \\
        & & & $(\delta^{\Obs}(\firstbasisstate,i\sigma){\uparrow} \lor \delta^{\Obs}(\secondbasisstate,\sigma){\uparrow})$\Tstrut &  \\ \cline{2-5}
        & \multirow{2}{=}{\emph{prioritized} \emph{promotion}}\Tstrut & $\firstfrontierstate \in \frontier$ & $\accessT(\firstfrontierstate) \in P^{\R}, \forall \firstbasisstate \in \basis.~\firstbasisstate \apart \firstfrontierstate$ & $\firstfrontierstate \in \basis$ \\ 
        &  &  &  &  \\ \hline
        \parbox[t]{2mm}{\multirow{7}{*}{\rotatebox[origin=c]{90}{Sec.~\ref{sec:ms},~\ref{sec:mr_ps}}}} & \multirow{2}{=}{\emph{match} \emph{separation}}\Tstrut & $\firstbasisstate,\secondbasisstate \in \basis$, & $ \delta^{\Obs}(\firstbasisstate,i) = \firstfrontierstate \in \frontier, \neg(\firstfrontierstate \apart \secondbasisstate)$, $\delta^{\R}(\firstrefbasisstate,i)=\secondrefbasisstate$ & $\firstfrontierstate \apart \secondbasisstate~\lor$\\
        & & $\firstrefbasisstate \in Q^{\R}, i \in I$ & $\neg(\exists \thirdbasisstate\in \basis$ s.t. $\secondrefbasisstate \matches \thirdbasisstate), \firstrefbasisstate \matches \firstbasisstate$ & $(\firstrefbasisstate \notmatches \firstbasisstate \land \firstfrontierstate \apart \secondrefbasisstate)$ \\ \cline{2-5}
        & \multirow{2}{=}{\emph{match} \emph{refinement}}\Tstrut & $\firstbasisstate \in \basis,$ & $\firstrefbasisstate \matches \firstbasisstate, \secondrefbasisstate \matches \firstbasisstate,$ & $\firstrefbasisstate \notmatches \firstbasisstate \lor \secondrefbasisstate \notmatches \firstbasisstate$ \\ 
        &  & $\firstrefbasisstate, \secondrefbasisstate \in Q^{\R}$ & $\sigma = \mathsf{sep}(\firstrefbasisstate,\secondrefbasisstate)$ &  \\ \cline{2-5}
        & \multirow{2}{=}{\emph{prioritized} \emph{separation}}\Tstrut & $\firstfrontierstate \in \frontier,$ & $\neg(\firstfrontierstate \apart \secondbasisstate), \neg(\firstfrontierstate \apart \thirdbasisstate),\exists i \in I$ s.t. $ \delta^{\Obs}(\firstbasisstate,i) = \firstfrontierstate,$  & $\firstfrontierstate \apart \thirdbasisstate~\lor$ \\
        & & $\secondbasisstate, \thirdbasisstate \in \basis$ & $\sigma \vdash \secondbasisstate \apart \thirdbasisstate, \sigma \in \cup_{\firstrefbasisstate \matches \firstbasisstate} W_{\delta^{\R}(\firstrefbasisstate,i)}$ &  $\firstfrontierstate \apart \secondbasisstate$\\ 
        \hline
    \end{tabular}}
    \label{tab:rules}
\end{table}

\begin{example}
    Suppose we learn $\R_1$ from Fig.~\ref{fig:runningexample}. \lsharp{} 
    applies the \emph{extension} rule twice, resulting in $\Obs$ as in Fig.~\ref{fig:tree1}.  
    States $t_1$ and $t_2$ are identified with $t_0$ because there is only one basis state. Next, \lsharp{} applies the \emph{equivalence} rule using hypothesis $\H$ (Fig.~\ref{fig:hyp1}). Counterexample $aac$ distinguishes $\H$ from $\R_1$. This sequence is added to $\Obs$ and processed further by posing OQ $ac$ in the \emph{equivalence} rule.
    Observations $ac$ and $aac$ show that the states accessed with $\varepsilon$, $a$ and $aa$ are pairwise apart. States $t_2$ and $t_3$ are added to the basis using the \emph{promotion} rule. 
    Next, \lsharp{} poses OQ $aaa$ during the \emph{extension} rule. 
    To identify all frontier states, \lsharp{} may use $ac \vdash t_2 \apart t_3$, $ac \vdash t_0 \apart t_2$ and $c \vdash t_0 \apart t_3$. 
    Fig.~\ref{fig:tree2} shows one possible observation tree $\Obs'$ after applying the \emph{separation} rule multiple times. Next, the \emph{equivalence} rule constructs hypothesis $\H'$ (Fig.~\ref{fig:hyp2}) from $\Obs'$ and \lsharp{} terminates because $\H'$ and $\R_1$ are equivalent.
\end{example}

\subsection{Rebuilding in $L^{\#}$} \label{sec:rebuild}
In this subsection, we combine rebuilding from~\cite{DBLP:conf/ifm/DamascenoMS19,DBLP:conf/birthday/FerreiraHS22} with \lsharp{} and implement this using two rules: \emph{rebuilding} and \emph{prioritized promotion}, see also Table~\ref{tab:rules}.
Both rules depend on a reference model $\R$, which is a complete Mealy machine, with a possibly different alphabet than the SUL $\S$. More precisely, these rules depend on a prefix-closed and minimal state cover $P^{\R}$ and a separating family $\{W_q\}^{\R}$ computed on $\R|_{I^{\S}}$ for maximal overlap with $\S$. 
The separating family can be computed with partition refinement~\cite{DBLP:conf/lata/SmetsersMJ16}. We fix $\mathsf{sep}(p,p')$ with $p, p' \in Q^{\R}$ to be a unique sequence from $W_{p} \cap W_{p'}$ such that $\mathsf{sep}(p,p') \vdash p \apart p'$.
Below, we use $\firstbasisstate$ for states in $\basis$, $\firstfrontierstate$ for states in $\frontier$ and $\firstrefbasisstate$ for states in~$Q^{\R}$. In App.\ifthenelse{\boolean{arxivversion}}{~\ref{app:A}}{~A~of~\cite{Kruger2024AdaptiveArxiv}}, we depict the scenarios in the observation tree and reference model required for the new rules to be applicable.

\begin{ruleenv}{R}{Rebuilding}
 Let $\firstbasisstate \in \basis$, $i \in I$ and suppose $\delta^{\Obs}(\firstbasisstate,i) \notin \basis$.
The aim of the \emph{rebuilding} rule is to show apartness between $\delta^{\Obs}(\firstbasisstate,i)$ and a basis state $\secondbasisstate$, using the state cover and separating family from $\R$. 
The rebuilding rule is applicable when $\accessT(\firstbasisstate)$ and $\accessT(\firstbasisstate)i$ are in $P^{\R}$. 
If $\accessT(\secondbasisstate) \in P^{\R}$ then there exists a sequence $\sigma$ such that $\sigma = \mathsf{sep}\bm{(}\delta^{\R}(\accessT(\firstbasisstate)i),\delta^{\R}(\accessT(\secondbasisstate))\bm{)}$. We pose OQs $\accessT(\firstbasisstate)i\sigma$ and $\accessT(\secondbasisstate)\sigma$.
\end{ruleenv}

\begin{lemma} \label{lem:rebuild}
    Suppose $\accessT(\secondbasisstate) \in P^{\R}$ for all $\secondbasisstate \in B$. 
       Consider $\firstbasisstate \in \basis$, $i \in I$ such that
        $\delta^{\Obs}(\firstbasisstate,i) \notin \basis$ and $\accessT(\firstbasisstate)i \in P^{\R}$.
           If for all $q' \in B$ it holds that $\mathsf{sep}\bm{(}\delta^{\R}(\accessT(\firstbasisstate)i),\delta^{\R}(\accessT(\secondbasisstate))\bm{)} \vdash \delta^{\S}(\accessT(\firstbasisstate)i) \apart \delta^{\S}(\accessT(\secondbasisstate))$,
            then after applying the \emph{rebuilding} rule for $\firstbasisstate$, $i$ and all $\secondbasisstate \in \basis$ with $\neg(\secondbasisstate \apart \delta^{\Obs}(\firstbasisstate,i))$, state $\delta^{\Obs}(\firstbasisstate,i)$ is isolated.
    \end{lemma}

\noindent If a state is isolated, it can be added to the basis using the \emph{promotion} rule.

\begin{ruleenv}{PP}{Prioritized promotion}
Like (regular) promotion, \emph{prioritized promotion} extends the basis. However, \emph{prioritized promotion} only applies to states $r$ with $\accessT(r) \in P^{\R}$. This enforces that the access sequences for basis states are in $P^{\R}$ as often as possible, enabling the use of the \emph{rebuilding} rule.
\end{ruleenv}

\begin{example} 
    Consider reference $\R_1$ and SUL $\S$ from Fig.~\ref{fig:runningexample}.
    We learn the orange states similarly as described in~Sec.~\ref{sec:overview}: We apply the \emph{rebuilding} rule with $\accessT(q) = \varepsilon, \accessT(q') = \varepsilon, i = c$ which results in OQs $cac$ and $ac$. Next, we promote $\delta^{\Obs}(c)$ with the \emph{prioritized promotion} rule. We apply the \emph{rebuilding} rule with $\accessT(q) = c, \accessT(q') = c$ and $i = a$ which results in OQs $cac$ (already present in $\Obs$) and $cc$. Lastly, we promote $\delta^{\Obs}(ca)$ with \emph{prioritized promotion}.
\end{example}

The overlap between $\S$ and $P^\R$ and $\{W_q\}^\R$ determines how many states of $\S$ can be discovered via rebuilding. The statement follows from Lemma~\ref{lem:rebuild} above.

\begin{theorem} \label{thm:completerebuilding}
    If $q_0^{\R}$ matches $q_0^{\S}$ and $\Obs$ only contains a root $q_0^{\Obs}$, then after applying only the \emph{rebuilding} and \emph{prioritized promotion} rules until they are no longer applicable, the basis consists of $n$ states where $n$ is the number of equivalence classes (w.r.t. language equivalence) in the reachable part of $\S|_{I^{\R}}$. 
\end{theorem}

\begin{corollary}
    Suppose we learn SUL $\S$ with reference $\S$. Using the \emph{rebuilding} and \emph{prioritized promotion} rules, we can add all reachable states in $\S$ to the basis.
\end{corollary}

%% file: figures/M1Lexample.tex
\begin{figure}[t!]
     \centering
     \begin{minipage}[b]{0.18\textwidth}
     \begin{center}
          \resizebox{.6\width}{!}{
          
          \includegraphics{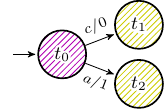}
          }
          \subcaption{$\Obs$}
          \label{fig:tree1}
          \vspace{0.5cm}
          \resizebox{.6\width}{!}{
               
               
          \includegraphics{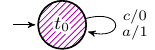}
               }
               \subcaption{$\Hyp$}
               \label{fig:hyp1}
     \end{center}
\end{minipage}
\begin{minipage}[b]{0.55\textwidth}
     \begin{center}
          \resizebox{.6\width}{!}{
          \includegraphics{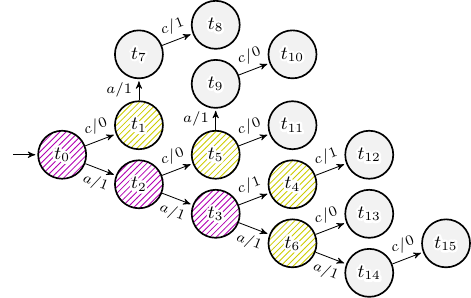}}
     \subcaption{$\Obs'$}	
     \label{fig:tree2}
     \end{center}
     \end{minipage}
     \begin{minipage}[b]{0.22\textwidth}
          \begin{center}
          \resizebox{.6\width}{!}{

          \includegraphics{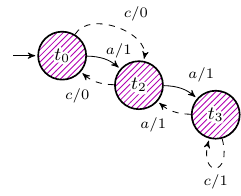}}
          \subcaption{$\H'$}	
          \label{fig:hyp2}
          \end{center}
     \end{minipage}
     \caption{Observation trees and hypotheses generated while learning $\R_1$ with \lsharp. Basis states are displayed in pink and frontier states in yellow. }	
\end{figure}

%% file: statematching.tex
\section{$L^{\#}$ using State Matching} \label{sec:statematch}
In this section, we describe another way to reuse information from references, called state matching, which is independent of the state cover. 
First, we present a version of state matching using the matching relation  (\matches) from Def.~\ref{def:matching} and then we weaken this notion to approximate state matching. 

\subsection{State Matching} \label{sec:qualstatematch}
\input{figures/MLmatching.tex}
With state matching, the learner maintains the matching relation $\matches$ between basis states and reference model states during learning. In the implementation, before applying a matching rule, the matching is updated based on the OQs asked since the previous match computation. We present two key rules here and an optimisation in the next subsection. 

\begin{ruleenv}{MS}{Match separation} \label{sec:ms}
This rule aims to show apartness between the frontier and a basis state using separating sequences from the reference separating family.
Let $\firstbasisstate$, $\secondbasisstate \in \basis$, $\firstfrontierstate \in \frontier$ with $\delta^{\Obs}(\firstbasisstate,i) = \firstfrontierstate$ for some $i \in I$, and $\firstrefbasisstate,\secondrefbasisstate \in Q^\R$.
Suppose that $\delta^{\R}(\firstrefbasisstate,i) = \secondrefbasisstate$, $\neg(\firstfrontierstate \apart \secondbasisstate)$, $\firstrefbasisstate \matches \firstbasisstate$ and $\secondrefbasisstate$ does not match any basis state.
In particular, there exists some separating sequence $\sigma$ for $\secondrefbasisstate \apart \secondbasisstate$.   
The \emph{match separation} rule poses OQ $\access(\firstbasisstate)i\sigma$ to either show $\firstfrontierstate \apart \secondbasisstate$ or $\firstbasisstate \notmatches \firstrefbasisstate$ and $\firstfrontierstate \apart \secondrefbasisstate$.
\end{ruleenv}

\begin{example}
    Suppose we learn $\S$ using $\R_2$ from Fig.~\ref{fig:runningexample}. After applying the \emph{extension} rule three times, we get $\Obs_0$  (Fig.~\ref{fig:mtree1}).
    State $t_0$ matches $p_3$ as their outputs coincide on sequences from alphabet $I^{\S} \cap I^{\R_2} = \{a,b\}$.
     State $p_3$ transitions to the unmatched state $p_0$ with input $a$. 
    The \emph{match separation} rule conjectures $t_1$ may match $p_0$ which implies $t_1 \apart t_0$. We use OQ $\access(t_1)a$ to test this conjecture and indeed find that $t_1$ can be added to the basis using \emph{promotion}. 
    \end{example}

\begin{lemma} \label{lem:statematch}
    We fix $\firstrefbasisstate \in Q^{\R}$, $\firstbasisstate \in \basis$, $i \in I$ and $\delta^{\Obs}(\firstbasisstate,i)=\firstfrontierstate\in \frontier$. Suppose $\delta^{\S}(\accessT(\firstbasisstate)) \matches \firstrefbasisstate$. If $\delta^{\R}(p,i) \notmatches \secondbasisstate$ for all $\secondbasisstate \in \basis$, then after applying the \emph{match separation} rule with $\firstbasisstate, \firstrefbasisstate, i$ for all $\secondbasisstate \in \basis$ with $\neg(\secondbasisstate \apart \firstfrontierstate)$, state $\firstfrontierstate$ is isolated.
\end{lemma}

\begin{ruleenv}{MR}{Match refinement}
    Let $\firstbasisstate \in \basis$ and $\firstrefbasisstate, \secondrefbasisstate \in Q^{\R}$. Suppose $\firstbasisstate$ matches both $\firstrefbasisstate$ and $\secondrefbasisstate$ and let $\sigma=\mathsf{sep}(\firstrefbasisstate,\secondrefbasisstate)$. 
    The \emph{match refinement} rule poses OQ $\access(\firstbasisstate)\sigma$ resulting in $\firstbasisstate$ no longer being matched to  $\firstrefbasisstate$ or $\secondrefbasisstate$.
\end{ruleenv}

\begin{example}
    Suppose we continue learning $\S$ using $\R_2$ from observation tree $\Obs_1$ (Fig.~\ref{fig:mtree1}). State $t_1$ matches both $p_0$ and $p_1$. After posing OQ $\access(t_1)bb$ where $bb \vdash p_0 \apart p_1$, $t_1$ no longer matches $p_1$.
\end{example}

If the initial state of SUL $\S$ is language equivalent to some state in the reference model, then we can discover all reachable states in $\S$ via state matching and \lsharp{} rules. The statement uses Lemma~\ref{lem:statematch} above.

\begin{theorem} \label{thm:statematching}
    Suppose we have reference $\R$ and SUL $\S$ equivalent to $\R$ but with a possibly different initial state. Using only the \emph{match refinement}, \emph{match separation}, \emph{promotion} and \emph{extension} rules, we can add $n$ states to the basis where $n$ is the number of equivalence classes (w.r.t. language equivalence) in the reachable part of $\S$.
\end{theorem} 

\subsection{Optimised Separation using State Matching} \label{sec:mr_ps}
In this subsection, we add an optimisation rule \emph{prioritized separation} that uses the matching to guide the identification of frontier states. First, we highlight the differences between \emph{prioritized separation} and the previous separation rules.
Both \emph{match separation} and \emph{prioritized separation} require that $\firstfrontierstate \matches \firstrefbasisstate$ for $\firstfrontierstate \in \frontier$ and $\firstrefbasisstate \in Q^{\R}$. The aim of \emph{match separation} is to isolate $\firstfrontierstate$ and requires that $\firstrefbasisstate$ does not match any basis state. Instead, the aim of \emph{prioritized separation} is to guide the identification of $\firstfrontierstate$ using the state identifier for a $\firstrefbasisstate$ matched with a basis state. The \emph{prioritized separation} rule is also different from the \emph{separation} rule (Sec.~\ref{sec:lsharp}) which randomly selects $\firstbasisstate, \secondbasisstate \in \basis$ to separate $\firstfrontierstate$ from $\firstbasisstate$ or $\secondbasisstate$. 

\begin{ruleenv}{PS}{Prioritized separation}
 The \emph{prioritized separation} rule uses the matching to find a separating sequence from the reference model that is expected to separate a frontier state from a basis state.
 Let $\secondbasisstate, \thirdbasisstate \in \basis$ and $\firstfrontierstate \in \frontier$. Suppose $\firstfrontierstate$ is not apart from $\secondbasisstate$ and $\thirdbasisstate$ and $\sigma \vdash \secondbasisstate \apart \thirdbasisstate$. 
If $\sigma$ is in $\{W_\firstrefbasisstate\}^{\R}$ of a reference model state $\firstrefbasisstate$ that matches $\firstfrontierstate$, the \emph{prioritized separation} rule poses OQ $\access(\firstfrontierstate)\sigma$ resulting in $\firstfrontierstate$ being apart from $\secondbasisstate$ or $\thirdbasisstate$\footnote{The precise specification is more involved, as the learner only keeps track of the match relation on $\basis \times Q^{\R}$.}.
\end{ruleenv}

\begin{example}
    Suppose we learn $\S$ using $\R_1$ from Fig.~\ref{fig:runningexample}. Assume we have discovered all states in $\S$ and want to identify $\delta^{\Obs}(ca,c) \in F$, which is currently not apart from any basis state. The \emph{prioritized separation} rule can only be applied with basis states $\secondbasisstate, \thirdbasisstate \in \basis$ such that $c \vdash \secondbasisstate \apart \thirdbasisstate$, as $c$ is the only sequence in the state identifier of $r_2$ which is the state that matches $\delta^{\Obs}(ca,c)$. From the sequences $\{bb, ac, c\}$ possibly used by \lsharp{}, only $c$ immediately identifies $\delta^{\Obs}(ca,c)$.
\end{example}

\subsection{Approximate State Matching} \label{sec:approxstatematch}
In this subsection, we introduce an approximate version of matching, by quantifying matching via a \emph{matching degree}.
\newcommand{\widef}[1]{\mathsf{WI}(#1)}
    Let $\Obs$ be a tree and $\R$ be a (partial) Mealy machine. 
    Let $I = I^{\Obs} \cap I^{\R}$. We define $\widef{\firstbasisstate} = \{ (w,i) \in I^* \times I \mid \delta^{\Obs}(\firstbasisstate,wi)\converges \}  $ 
    as prefix-suffix pairs that are defined from $\firstbasisstate \in Q^{\Obs}$ onwards.
    Then, we define the matching degree $\mathsf{mdeg}: Q^{\Obs} \times Q^{\R} \to \mathbb{R}$ as 
    \begin{equation*}\mathsf{mdeg}(\firstbasisstate,\firstrefbasisstate) = 
        \frac{ \left| \{ (w,i) \in \widef{\firstbasisstate} \mid \lambda^{\Obs}\bigg(\delta^{\Obs}(\firstbasisstate,w),i\bigg) = \lambda^{\R}\\\bigg(\delta^{\R}(\firstrefbasisstate,w),i\bigg) \} \right|}{ \left|\widef{\firstbasisstate} \right|}.
    \end{equation*}

\begin{example} 
    Consider $t_1$ from $\Obs_2$ (Fig.~\ref{fig:mtree3}) and $p_0$, $p_1$ from $\R_2$ (Fig.~\ref{fig:runningexample}). We derive $\widef{t_1} = \{ (\varepsilon,a), (\varepsilon,b), (b,a), (b,b), (bb,b) \}$ from $\Obs_2$ where $I = I^{\Obs_2} \cap I^{\R_2} = \{ a,b \}$. On these pairs, all the suffix outputs for $p_0$ and $t_1$ are equivalent, $\mathsf{mdeg}(t_1,p_0) = \nicefrac{5}{5} = 1$. The matching degree between $t_1$ and $p_1$ is only $\nicefrac{3}{5}$ because $\lambda^{\R_2}(p_1,bbb) = 120 \neq 112 = \lambda^{\Obs}(t_1,bbb)$ which impacts pairs $(b,b)$ and $(bb,b)$.
\end{example}

A state $\firstbasisstate$ in an observation tree $\Obs$ \emph{approximately matches} a state $\firstrefbasisstate \in Q^{\R}$, written $\firstbasisstate \approxmatches \firstrefbasisstate$, if there does not exist a $\secondrefbasisstate \in Q^{\R}$ such that $\mathsf{mdeg}(\firstbasisstate,\secondrefbasisstate) > \mathsf{mdeg}(\firstbasisstate,\firstrefbasisstate)$.

\begin{lemma} \label{lem:approx1}
	For any $\firstbasisstate \in Q^{\Obs}, \firstrefbasisstate \in Q^{\R}$: $\mathsf{mdeg}(\firstbasisstate,\firstrefbasisstate) = 1$ implies $\firstbasisstate \matches \firstrefbasisstate$.
\end{lemma}

We define rules \emph{approximate match separation} (\textbf{AMS}), \emph{approximate match refinement} (\textbf{AMR}) and \emph{approximate prioritized separation} (\textbf{APS}) that represent the approximate matching variations of \emph{match separation}, \emph{match refinement} and \emph{prioritized separation} respectively. These rules have weaker preconditions and postconditions, see \ifthenelse{\boolean{arxivversion}}{Table~\ref{tab:approxrules} in App~\ref{app:A}}{Table~3 in App~A~of~\cite{Kruger2024AdaptiveArxiv}}.

%% file: figures/MLmatching.tex
\begin{figure}[t!]
     \centering
     \begin{minipage}[b]{0.26\textwidth}
     \begin{center}
          \resizebox{.7\width}{!}{
          
               	
          \includegraphics{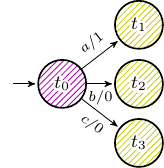}
          }
          \subcaption{Observation tree $\Obs_0$}
          \label{fig:mtree1}
     \end{center}
     \end{minipage}
     \hspace*{0.2cm}
     \begin{minipage}[b]{0.27\textwidth}
          \begin{center}
          \resizebox{.7\width}{!}{
          
               	
          \includegraphics{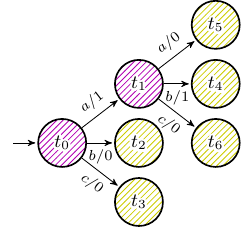}}
          \subcaption{Observation tree $\Obs_1$}
          \label{fig:mtree2}
     \end{center}
\end{minipage}
\begin{minipage}[b]{0.43\textwidth}
     \begin{center}
          \resizebox{.7\width}{!}{
          
          \includegraphics{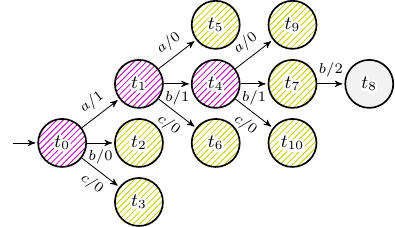}}
          \subcaption{Observation tree $\Obs_2$}
          \label{fig:mtree3}
     \end{center}
     \end{minipage}
     \caption{Observation trees generated while learning $\S$ with $\R_2$.}	
\end{figure}

%% file: adaptivelsharp.tex
\section{Adaptive $L^{\#}$} \label{sec:alsharp}
The rebuilding, state matching and \lsharp{} rules described in Table~\ref{tab:rules} are ordered and combined into one adaptive learning algorithm called \emph{adaptive} \lsharp{} (written \adaptivelsharp). A non-ordered listing of the rules can be found in \ifthenelse{\boolean{arxivversion}}{Algorithm~\ref{alg:extendedlsharp} in App.~\ref{app:A}}{Algorithm~1 in App.~A~of~\cite{Kruger2024AdaptiveArxiv}}. We use the abbreviations for the rules defined in previous sections.

\begin{definition}
    The \adaptivelsharp{} algorithm repeatedly applies the rules from Table~\ref{tab:rules} (see Algorithm~\ifthenelse{\boolean{arxivversion}}{\ref{alg:extendedlsharp}}{1}), with the following ordering: \textbf{Ex}, \textbf{APS}, (\textbf{S} if \textbf{APS} was not applicable), \textbf{P}, if the observation tree is adequate we try \textbf{AMR}, \textbf{AMS}, \textbf{Eq}.
	The algorithm starts by applying \textbf{R} and \textbf{PP} until they are no longer applicable; these rules are not applied anymore afterwards.
\end{definition} 

Similar to \lsharp, the correctness of \adaptivelsharp{} amounts to showing termination because the algorithm can only terminate when the teacher indicates that the SUL and hypothesis are equivalent. We prove termination of \adaptivelsharp{} by proving that each rule application lowers a ranking function. The necessary ingredients for the ranking function are derived from the post-conditions of Table~\ref{tab:rules}.

\begin{theorem} \label{thm:quantitativeComplexity}
    \adaptivelsharp learns the correct Mealy machine within $\bigO(kn^2 + kno + no^2 + n \log m)$ output queries and at most $n-1$ equivalence queries where $n$ is the number of equivalence classes for $\S$, $o$ is the number of equivalence classes for $\R$, $k$ is the number of input symbols and $m$ the length of the longest counterexample.
\end{theorem}

%% file: multiplereferences.tex
\section{Adaptive Learning with Multiple References}\label{sec:multiple}
Let $\X$ be a finite set of complete reference models with possibly different alphabets. Assume each reference model $\R \in \X$ has a state cover $P^{\R}$ and separating family $\{W_q\}^{\R}$. We adapt the arguments for the \adaptivelsharp{} algorithm to represent the state cover and separating family for the set of reference models.

\myparagraph{State cover.} We initialize the \adaptivelsharp{} algorithm with the union of the state cover of each reference model, i.e., $\cup_{\R \in \X} P^{\R}$.
To reduce the size of $P^{\X}$, the state cover for each reference model is computed using a fixed ordering on inputs.

\myparagraph{Separating family.} We combine the separating families for multiple reference models using a stronger notion of apartness, called \emph{total apartness}, which also separates states based on whether inputs are defined. When changing the alphabet of a reference model to the alphabet of the SUL, as is done when computing the separating family, the reference model may become partial.
If states from different reference models behave the same on their common alphabet but their alphabets contain different inputs from the SUL, we still want to distinguish the reference models based on which inputs they enable.

\begin{definition}
    Let $\M_1, \M_2$ be partial Mealy machines and $p \in Q^{\M_1}, q \in Q^{\M_2}$. We say $p$ and $q$ are \emph{total apart}, written $p \totalapart q$, if $p \apart q$ or there exists $\sigma \in (I^{\M_1} \cap I^{\M_2})^*$ such that either $\delta^{\M_1}(p,w){\uparrow}$ or  $\delta^{\M_2}(q,w){\uparrow}$ but not both.
\end{definition}

We use \emph{total apartness} to define a \emph{total state identifier} and a \emph{total separating family}. This definition is similar to Def.~\ref{def:sepfam} but $\apart$ is be replaced by $\totalapart$. We combine the multiple reference models into a single one with an arbitrary initial state, compute the \emph{total separating family} and use this to initialize \adaptivelsharp{}.

\begin{example}
    A total separating family for $\X = \{ \R_1, \R_2 \}$ and alphabet $I^{\S}$ is
    $W_{p_0}=W_{p_1}=\{c,b,bb\}, W_{p_2}=W_{p_3}=\{c,b\}, W_{r_0}=W_{r_1}=\{c,ac\}, W_{r_2}=\{c\}$.
\end{example}

We add an optimisation to \adaptivelsharp{} that only chooses $\firstrefbasisstate$ and $\secondrefbasisstate$ from the same reference model during rebuilding. Theorem~\ref{thm:quantitativeComplexity} can be generalized to this setting where $o$ represents the number of equivalence classes across the reference models.

%% file: experiments.tex
\section{Experimental Evaluation} \label{sec:experiments}
In this section, we empirically investigate the performance of our implementation of \adaptivelsharp. 
The source code and all benchmarks are available online\footnote{\url{https://gitlab.science.ru.nl/lkruger/adaptive-lsharp-learnlib/}}\cite{KrugerZenodo}. 
We present four experiments to answer the following research questions:
\begin{description}[topsep=2pt, partopsep=0pt]
\item[R1] What is the performance of \aaal{} algorithms, when \dots
\begin{description}
    \item[Exp 1] \dots learning models from a similar reference model?
    \item[Exp 2] \dots applied to benchmarks from the literature?
\end{description}
\item[R2] Can multiple references help \adaptivelsharp, when learning \dots
\begin{description}[topsep=0pt]
    \item[Exp 3] \dots a model from similar reference models?
    \item[Exp 4] \dots a protocol implementation from reference implementations?
\end{description}
\end{description}

\mysubsubsection{Setup.}
We implement \adaptivelsharp on top of the \lsharp{} LearnLib implementation\footnote{Obtained from \url{https://github.com/UCL-PPLV/learnlib.git}~\cite{DBLP:conf/birthday/FerreiraHS22}}.
We invoke conformance testing for the EQs, using the \texttt{random Wp method} from LearnLib with \texttt{minimal size}${=}3$ and \texttt{random length}${=}3$~\footnote{These hyperparameters are discussed in the LearnLib documentation, \url{learnlib.de}.}. 
We run all experiments with 30 seeds. 
We measure the performance of the algorithms based on the number of inputs sent to the SUL during both OQs and EQs: \emph{Fewer is better}.

\input{experiment_results/my_exp1_table.tex}

\mysubsubsection{Experiment 1.} 
We evaluate the performance of \adaptivelsharp{} against non-adaptive and adaptive algorithms from the literature, in particular  $L^*$~\cite{DBLP:journals/iandc/Angluin87}, \textit{KV}~\cite{DBLP:books/daglib/0041035}, and \lsharp~\cite{DBLP:conf/tacas/VaandragerGRW22} as well as \pdlstar \cite{DBLP:conf/ifm/DamascenoMS19} and (a Mealy machine adaptation of) \IKV{} \cite{DBLP:conf/birthday/FerreiraHS22}.
As part of an ablation study, we compare \adaptivelsharp with simpler variations which we refer to as \scalebox{0.88}{\lsharprebuild}, \scalebox{0.88}{\lsharpmatch}, \scalebox{0.88}{\lsharpapproxmatch}, \scalebox{0.88}{\lsharprebuildmatch}. The subscripts indicate which rules are added:\\ $R$: \textbf{R} + \textbf{PP}, $\matches$: \textbf{MS} + \textbf{MR} + \textbf{PS}, $\approxmatches$: \textbf{AMS} + \textbf{AMR} + \textbf{APS}.

We learn six models from the AutomataWiki benchmarks~\cite{DBLP:conf/birthday/NeiderSVK97} also used in~\cite{DBLP:conf/tacas/VaandragerGRW22}. We limit ourselves to six models because we mutate every model in 14 different ways (and for 30 seeds). The chosen models represent different types of protocols with varying number of states. 
We learn the mutated models using the original models, referred to as $\S$, as a reference. The mutations may add states, divert transitions, remove inputs, perform multiple mutations, or compose the model with a mutated version of the model. We provide details on the used models and mutations in \ifthenelse{\boolean{arxivversion}}{App.~\ref{app:E}.}{App.~E~of~\cite{Kruger2024AdaptiveArxiv}.}

\myparagraph{Results.}
Table~\ref{tab:my_exp1} shows for an algorithm (rows) and a mutation (columns) the total number of inputs ($\cdot 10^6$) necessary to learn all models, summed over all seeds\footnote{\pdlstar{} and \IKV{} do not support removing input inputs, relevant for mutation M7.}.
The \textcolor{teal}{highlighted} values indicate the best performing algorithm. 
We provide detailed pairwise comparisons between algorithms in \ifthenelse{\boolean{arxivversion}}{App.~\ref{app:E}.}{App.~E~of~\cite{Kruger2024AdaptiveArxiv}.}

\myparagraph{Discussion.}
First, we observe that \adaptivelsharp{} always outperforms non-adaptive learning algorithms, as is expected. 
By combining state matching and rebuilding, \adaptivelsharp{} mostly outperforms algorithms from the literature, with \IKV{} being competitive on some types of mutations. In $\textit{mut}_{9}(\S)$ we append $\S$ to $\textit{mut}_{13}(\S)$, \scalebox{0.88}{\lsharpmatch{}} outperforms \scalebox{0.88}{\lsharpapproxmatch{}} because \scalebox{0.88}{\lsharpapproxmatch{}} incorrectly matches $\textit{mut}_{13}(\S)$ states with states in $\S$, making it harder to learn the $\S$ fragment. 
\ifthenelse{\boolean{arxivversion}}{In the pairwise comparisons in App.~\ref{app:E}, we see that \adaptivelsharp{} performs much better on models \textit{GnuTLS},  \textit{OpenSSH} compared to other adaptive approaches. We conjecture that this effect occurs, as these models are hard to learn in general (high number of total inputs) and thus the potential benefit of \adaptivelsharp{} is higher. }{}

\mysubsubsection{Experiment 2.}
We evaluate \lsharp, \pdlstar, \IKV{} and \adaptivelsharp{} on benchmarks that contain reference models. \textit{Adaptive-OpenSSL}~\cite{DBLP:conf/nordsec/Ruiter16}, used in~\cite{DBLP:conf/ifm/DamascenoMS19}, contains models learned from different git development branches for the OpenSSL server side. \textit{Adaptive-Philips}~\cite{DBLP:conf/ifm/SchutsHV16} contains models representing some legacy code which evolved over time due to bug fixes and allowing more inputs.

\begin{figure}[t]%
    \begin{subfigure}[htb]{0.58\textwidth}
        \resizebox{.99\textwidth}{!}{
        \includegraphics{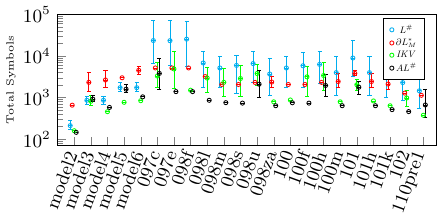}}
        \vspace*{-0.5cm}
        \caption{Averaged inputs for learning \textit{Adaptive-Philips} (starting with \textit{m}) and \textit{Adaptive-OpenSSL}. }
        \label{fig:res2}
    \end{subfigure}
    \hspace{0.02\textwidth}
    \begin{subfigure}[htb]{0.34\textwidth}
        \centering
        \subfloat[][Summed inputs in millions for learning some $\S$.]{\input{finalversionfigures/exp3_1.tex}\label{tab:exp5a}}%
        \newline
        \subfloat[][Summed inputs in millions for learning some mutated $\S$.]{\input{finalversionfigures/exp3_2.tex}\label{tab:exp5b}}
    \end{subfigure}
    \caption{Results Experiments 2 and 3.}
  \end{figure}

\myparagraph{Results.}
Fig.~\ref{fig:res2} shows the mean total number of inputs required for learning a model from the associated reference model, depicting the $5^{\text{th}}-95^{\text{th}}$ percentile (line) and average (mark) over the seeds.

\myparagraph{Discussion.}
We observe that \lsharp{} and \pdlstar{} perform worse than \adaptivelsharp{}. \adaptivelsharp{} often outperforms \IKV{} by a factor 2-4, despite that these models are relatively small and thus easy to learn. 

\mysubsubsection{Experiment 3.}
We evaluate \adaptivelsharp{} with one or multiple references on the models used in Experiment 1. We either (1) learn $\S$ using several mutations of $\S$ or (2) learn a mutation that represents a combination of the $\S$ and $\textit{mut}_{13}(\S)$. 

\myparagraph{Results.}
Tables~\ref{tab:exp5a}, \ref{tab:exp5b} show for every type of SUL (rows) and every set of references (columns) the total number of inputs ($\cdot 10^6$) necessary to learn all models, summed over all seeds. \textcolor{teal}{Highlighted} values indicate the best performing set of references. Column $\{\S\}$ in Table~\ref{tab:exp5b} corresponds to values in row \adaptivelsharp{} of Table~\ref{tab:my_exp1}; they are added in Table~\ref{tab:exp5b} for clarity.

\myparagraph{Discussion.}
We observe that using multiple references outperforms using one reference, as is expected. We hypothesize that learning with reference $\textit{mut}_{13}(\S)$ instead of $\S$ often leads to an increase in total inputs because $\textit{mut}_{13}(\S)$ is less complex due to the random transitions. Therefore, discovering states belonging to the $\S$ fragment in $\textit{mut}_{8}(\S)$, $\textit{mut}_{9}(\S)$ and $\textit{mut}_{14}(\S)$ becomes more difficult.

\mysubsubsection{Experiment 4.}
We evaluate the performance of \adaptivelsharp{} with one or multiple references on learning \textit{DTLS} and \textit{TCP} models from AutomataWiki\footnote{References represent related models instead of previous models as in Experiment 2.}. We consider seven \textit{DTLS} implementations selected to have the same key exchange algorithm and certification requirement. We consider three \textit{TCP} client implementations. 

\begin{figure}[t]
    \centering
    \begin{subfigure}[b]{0.693\textwidth}
        \includegraphics{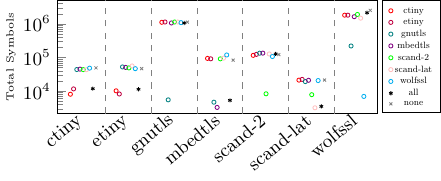}
    \end{subfigure}
    \begin{subfigure}[b]{0.297\textwidth}
        \includegraphics{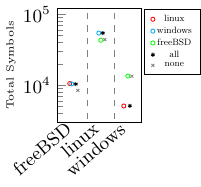}
    \end{subfigure}
    \vspace{-0.7cm}
    \caption{Averaged inputs for learning $\S$ with multiple references.}
    \label{fig:res4}
\end{figure}

\myparagraph{Results.}
Fig.~\ref{fig:res4} shows the required inputs for learning $\S$ (x-axis) with only the reference model indicated by the colored data point, averaged over the seeds. For each \textit{DTLS} model, we included learning $\S$ with the $\S$ as a reference model. The $*$ mark indicates using all models except the $\S$ as references, the $\times$ mark indicates using no references, e.g., non-adaptive \lsharp{}. 

\myparagraph{Discussion.}
We observe that using all references except $\S$ usually performs as well as the best performing reference model that is distinct from $\S$. In \textit{scand-lat}, using a set of references outperforms single reference models, almost matching the performance of learning $\S$ with $\S$ as a reference. 

%% file: experiment_results/my_exp1_table.tex
\renewcommand{\arraystretch}{1.2}
\begin{table}[t]
\centering
\caption{Summed inputs in millions for learning the mutated models with the original models.}
\resizebox{.99\textwidth}{!}{
\begin{tabular}{lp{1.2cm}p{1.2cm}p{1.2cm}p{1.2cm}p{1.2cm}p{1.2cm}p{1.2cm}p{1.2cm}p{1.2cm}p{1.2cm}p{1.2cm}p{1.2cm}p{1.2cm}p{1.2cm}}
\toprule
Algorithm & $\textit{mut}_1$ & $\textit{mut}_2$ & $\textit{mut}_3$ & $\textit{mut}_4$ & $\textit{mut}_5$ & $\textit{mut}_6$ & $\textit{mut}_7$ & $\textit{mut}_8$ & $\textit{mut}_9$ & $\textit{mut}_{10}$ & $\textit{mut}_{11}$ & $\textit{mut}_{12}$ & $\textit{mut}_{13}$ & $\textit{mut}_{14}$ \\
\midrule
$L^*$ & 115.2 & 24.2 & 49.4 & 69.7 & 78.7 & 60.5 & 50.7 & 132.9 & 294.2 & 36.8 & 52.5 & 38.0 & 18.3 & 301.9 \\
\textit{KV} & 123.5 & 17.8 & 49.6 & 60.1 & 68.9 & 58.7 & 44.9 & 103.7 & 244.3 & 25.5 & 28.7 & 28.0 & 7.5 & 253.6 \\
\lsharp & 101.7 & 14.3 & 50.0 & 49.2 & 73.0 & 58.7 & 39.9 & 100.1 & 313.9 & 25.4 & 38.9 & 28.0 & 8.0 & 234.9 \\
\midrule \pdlstar~\cite{DBLP:conf/ifm/DamascenoMS19} & 132.7 & 19.8 & 22.5 & 25.0 & 32.7 & 26.0 & - & 178.0 & 375.0 & 24.7 & 25.4 & 44.1 & 8.9 & 256.3 \\
\IKV~\cite{DBLP:conf/birthday/FerreiraHS22} & 114.8 & 18.6 & 1.6 & 2.4 & 0.9 & {\textcolor{teal}{ 0.8}} & - & {\textcolor{teal}{ 56.6}} & 373.9 & 11.0 & 2.1 & 1.1 & 5.8 & {\textcolor{teal}{ 7.0}} \\\midrule
\adaptivelsharp{} (\textcolor{red!60!black}{new!}) & 1.2 & {\textcolor{teal}{ 0.5}} & {\textcolor{teal}{ 1.5}} & {\textcolor{teal}{ 0.8}} & {\textcolor{teal}{ 0.8}} & 0.8 & {\textcolor{teal}{ 0.6}} & 68.1 & 141.1 & {\textcolor{teal}{ 1.4}} & {\textcolor{teal}{ 1.3}} & {\textcolor{teal}{ 0.8}} & {\textcolor{teal}{ 1.9}} & 7.2 \\\midrule
\lsharprebuild{} (\textcolor{red!60!black}{new!}) & 101.7 & 12.3 & 1.7 & 9.4 & 1.1 & 7.9 & 0.7 & 68.2 & 306.1 & 12.6 & 2.8 & 1.7 & 6.4 & 7.9 \\
\lsharpmatch{} (\textcolor{red!60!black}{new!}) & {\textcolor{teal}{ 1.2}} & 0.5 & 3.5 & 5.2 & 9.1 & 7.2 & 0.7 & 63.0 & {\textcolor{teal}{ 36.8}} & 8.7 & 9.8 & 10.8 & 5.7 & 7.1 \\
\lsharpapproxmatch{} (\textcolor{red!60!black}{new!}) & 1.2 & 0.5 & 1.7 & 2.7 & 2.0 & 2.1 & 0.7 & 70.6 & 186.5 & 6.0 & 6.1 & 1.7 & 4.8 & 7.4 \\
\lsharprebuildmatch{} (\textcolor{red!60!black}{new!}) & 1.2 & {\textcolor{teal}{ 0.5}} & 1.5 & 0.8 & 1.0 & 0.8 & {\textcolor{teal}{ 0.6}} & 69.3 & 38.7 & 3.1 & 2.0 & 1.0 & 4.5 & 7.3 \\
\bottomrule
\end{tabular}
}
\label{tab:my_exp1}
\end{table}

%% file: finalversionfigures/exp3_1.tex
\resizebox{.9\textwidth}{!}{
\begin{tabular}{p{1.5cm}p{1.65cm}p{1.65cm}p{1.65cm}p{1.65cm}}
\toprule
& & & \{$\textit{mut}_{10}(\S)$, & \{$\textit{mut}_{12}(\S)$, \\ 
& & & ~$\textit{mut}_{10}(\S)$, & ~$\textit{mut}_{12}(\S)$, \\ 
SUL & \{$\textit{mut}_{10}(\S)$\} & \{$\textit{mut}_{12}(\S)$\} & ~$\textit{mut}_{10}(\S)$\} & ~$\textit{mut}_{12}(\S)$\} \\ 
\midrule
$\S$ & 33.1 & 52.7 & ~{\textcolor{teal}{17.4}} & ~22.9\\ 
\bottomrule\\ 
\end{tabular}
}

%% file: finalversionfigures/exp3_2.tex
\resizebox{.9\textwidth}{!}{
\begin{tabular}{p{1.5cm}p{2.2cm}p{2.2cm}p{2.2cm}}
\toprule
SUL & \{$\S$\} & \{$\textit{mut}_{13}(\S)$\} & \{$\S$,$\textit{mut}_{13}(\S)$\} \\
\midrule
$\textit{mut}_{8}(\S)$ & 68.1 & 96.3 & {\textcolor{teal}{25.7}}\\
$\textit{mut}_{9}(\S)$ & 141.1 & 263.0 & {\textcolor{teal}{35.3}}\\
$\textit{mut}_{14}(\S)$ & 7.2 & 212.1 & {\textcolor{teal}{3.2}} \\
\bottomrule\\
\end{tabular}
}
    

%% file: conclusion.tex
\section{Conclusion} \label{sec:conclusion}
We introduced the \emph{adaptive \lsharp} algorithm (\adaptivelsharp), a new algorithm for adaptive active automata learning that allows
to flexibly use domain knowledge in the form of (preferably similar) reference models and thereby aims to reduce the sample complexity for learning new models. 
Experiments show that the algorithm can lead to significant improvements over the state-of-the-art (Sec.~\ref{sec:experiments}). 

\subsubsection{Future work.} Approximate state matching is sometimes too eager and may mislead the learner, as happens for $\textit{mut}_9$ in Experiment~1 (Sec.~\ref{sec:experiments}). This may be addressed by only applying matching rules when the matching degree is above some threshold. It is currently unclear how to determine an appropriate threshold.

Further, adaptive methods typically perform well when the reference model and SUL are similar~\cite{DBLP:conf/fmics/HuistraMP18}. We would like to dynamically determine which (parts of) reference models are similar, and incorporate this in the rebuilding rule. 

Adaptive AAL allows the re-use of information in the form of a Mealy machine. 
Other sources of information that can be re-used in AAL are, for instance, system logs, realised by combining active and passive learning~\cite{DBLP:conf/wcre/YangASLHCS19, DBLP:journals/corr/abs-2209-14031}. 
An interesting direction of research is the development of a more general methodology
that allows the re-use of various forms of previous knowledge.

%% file: appendix.tex
\section{Additional Definition, Figure, Table and Algorithm} \label{app:A}
We define how to fold back an observation tree to a complete Mealy machine.

\begin{definition} \label{def:foldedback}
  Let $\Obs$ be an observation tree for SUL $\S$. If each basis state has a transition for every input and each frontier state is identified with a basis state, then $\Obs$ is \emph{folded back} to complete Mealy machine $\H = (\basis, I, O, q_0^{\Obs}, \delta^{\H}, \lambda^{\Obs})$ where for all $q \in \basis$ and $i \in I$:
  \[ \delta^{\H}(q,i) = \begin{cases}
    \delta^{\Obs}(q,i) & \text{ if } \delta^{\Obs}(q,i) \in \basis \\
    q' & \text{ if } \delta^{\Obs}(q,i) = r \in \frontier \text{ and } r \text{ is identified with } q' \in \basis \\
  \end{cases}
  \] 
\end{definition}

\noindent In Fig.~\ref{fig:rules}, we show the scenarios in the observation tree and the reference model necessary to apply the \emph{rebuilding}, \emph{match refinement}, \emph{match separation} and \emph{prioritized separation} rules. 
\input{figures/rules.tex}

\noindent In Algorithm~\ref{alg:extendedlsharp}, we list the rules used for \adaptivelsharp{} in a non-deterministic ordering.
\input{alglisting.tex}

\clearpage
\noindent Table~\ref{tab:approxrules} shows the pre- and postconditions of the approximate matching variations of the state matching rules.

\begin{table}[H] 
  \centering
  \caption{Approximate state matching rules with parameters, preconditions and postconditions. }
  \begin{tabular}{p{0.25cm}|p{1.8cm}||p{1.8cm}|l|p{2.2cm}}
      & Rule & Parameters & Precondition & Postcondition \\ \hline \hline
      \parbox[t]{2mm}{\multirow{8}{*}{\rotatebox[origin=c]{90}{Sec~\ref{sec:approxstatematch}}}} & \multirow{3}{=}{\emph{approximate} \emph{match} \emph{separation}}\Tstrut & $\firstbasisstate,\secondbasisstate \in \basis$, & $ \delta^{\Obs}(\firstbasisstate,i) = \firstfrontierstate \in \frontier, \neg(\firstfrontierstate \apart \secondbasisstate),$ & $\firstfrontierstate \apart \secondbasisstate~\lor~$\\
      & & $\firstrefbasisstate \in Q^{\R},$ & $\neg(\firstfrontierstate \apart \secondrefbasisstate),\firstrefbasisstate \approxmatches \firstbasisstate, \delta^{\R}(\firstrefbasisstate,i)=\secondrefbasisstate$, & $\firstfrontierstate \apart \secondrefbasisstate$\\
      & & $i \in I$ &  $\neg(\exists \thirdbasisstate\in \basis$ s.t. $\secondrefbasisstate \approxmatches \thirdbasisstate)$ & \\ \cline{2-5}
      & \multirow{3}{=}{\emph{approximate} \emph{match} \emph{refinement}}\Tstrut & $\firstbasisstate \in \basis,$ & $\firstrefbasisstate \approxmatches \firstbasisstate, \secondrefbasisstate \approxmatches \firstbasisstate,$ & $\delta^{\Obs}(\firstbasisstate,\sigma){\converges},$ \\ 
      &  & $\firstrefbasisstate, \secondrefbasisstate \in Q^{\R}$ & $\sigma = \mathsf{sep}(\firstrefbasisstate,\secondrefbasisstate), \delta^{\Obs}(\firstbasisstate,\sigma){\uparrow}$ &\\
      & & & & \\ \cline{2-5}
      & \multirow{3}{=}{\emph{approximate} \emph{prioritized} \emph{separation}}\Tstrut & $\firstfrontierstate \in \frontier,$ & $\neg(\firstfrontierstate \apart \secondbasisstate), \neg(\firstfrontierstate \apart \thirdbasisstate)$, & $\firstfrontierstate \apart \thirdbasisstate \lor \firstfrontierstate \apart \secondbasisstate$ \\
      & & $\secondbasisstate, \thirdbasisstate \in \basis$ & $\exists i \in I$ s.t. $ \delta^{\Obs}(\firstbasisstate,i) = \firstfrontierstate,$ &  \\ 
      & & & $\sigma \vdash \secondbasisstate \apart \thirdbasisstate, \sigma \in \cup_{\firstrefbasisstate \approxmatches \firstbasisstate} W_{\delta^{\R}(\firstrefbasisstate,i)}$ & \\ 
      \hline
  \end{tabular}
  \label{tab:approxrules}
\end{table}

\section{Proofs of Section~\ref{sec:lsharprebuild}}
\subsection*{Proof of Lemma~\ref{lem:rebuild}}
\begin{proof}
  Let $\firstbasisstate \in \basis$, $i \in I$ and $\sigma \in I^*$. Suppose 
  \begin{enumerate}[label=(\arabic*)]
    \item $\delta^{\Obs}(\firstbasisstate,i) \notin \basis$,
    \item $\accessT(\firstbasisstate)i \in P^{\R}$,
    \item For all $\secondbasisstate \in \basis$, $\accessT(\secondbasisstate) \in P^{\R}$,
    \item For all $\secondbasisstate \in \basis$, $\sigma \vdash \delta^{\S}(\accessT(\firstbasisstate)i) \apart \delta^{\S}(\accessT(\secondbasisstate))$, where we write $\sigma = \mathsf{sep}\bm{(}\delta^{\R}(\accessT(\firstbasisstate)i),\delta^{\R}(\accessT(\secondbasisstate))\bm{)}$ for conciseness. 
  \end{enumerate}
  We prove for all $\secondbasisstate \in \basis$, $\delta^{\Obs}(\firstbasisstate,i) \apart \secondbasisstate$ holds from either assumptions (1)-(4) or because the assumptions validate preconditions for the rebuilding rule and after applying the rule we find the required result. 
  Suppose we have a specific  $\secondbasisstate \in \basis$. If $\delta^{\Obs}(\firstbasisstate,i) \apart \secondbasisstate$ holds, we are done. From now, assume (5) $\neg(\delta^{\Obs}(\firstbasisstate,i) \apart \secondbasisstate)$.

    From (4) we derive that (6) $\delta^{\Obs}(\firstbasisstate,i\sigma){\uparrow}$ or $\delta^{\Obs}(\secondbasisstate,\sigma){\uparrow}$. Otherwise, $\delta^{\Obs}(\firstbasisstate,i\sigma){\converges}$ and $\delta^{\Obs}(\secondbasisstate,\sigma){\converges}$ which implies $\delta^{\Obs}(\firstbasisstate,i) \apart \secondbasisstate$ under assumption $\sigma \vdash \delta^{\S}(\accessT(\firstbasisstate)i) \apart \delta^{\S}(\accessT(\secondbasisstate))$. However, $\delta^{\Obs}(\firstbasisstate,i) \apart \secondbasisstate$ contradicts (5). 
    
    From assumptions (1)-(3),(5),(6), we know \emph{rebuilding} can be applied which leads to OQ $\accessT(\firstbasisstate)i\sigma$ and $\accessT(\secondbasisstate)\sigma$. After the OQs, we know $\delta^{\Obs}(\firstbasisstate,i\sigma){\converges}$ and $\delta^{\Obs}(\secondbasisstate,\sigma){\converges}$, combining this with $\sigma \vdash \delta^{\S}(\accessT(\firstbasisstate)i) \apart \delta^{\S}(\accessT(\secondbasisstate))$ proves that $\delta^{\Obs}(\firstbasisstate,i) \apart \secondbasisstate$.
  Thus, for every $\secondbasisstate \in \basis$, $\delta^{\Obs}(\firstbasisstate,i) \apart \secondbasisstate$, which is exactly the definition of \emph{isolated}.
\end{proof}

\subsection*{Proof of Theorem~\ref{thm:completerebuilding}}
\begin{proof}
    Let $n$ be the number of equivalence classes (w.r.t. language equivalence) in the reachable part of $\S|_{I^{\R}}$. We prove that whenever the basis does not contain $n$ elements, then there always exists an access sequence in $P^{\R}$ that leads to a state that can be isolated using the \emph{rebuilding} rule. Using recursive reasoning, this proves that the basis contains $n$ states whenever \emph{prioritized promotion} and \emph{rebuilding} are not applicable anymore. Let $\basis, \frontier, \Obs$ denote the current basis, frontier and observation tree. From the Theorem statement we know:
    \begin{enumerate}[label=(\arabic*)]
      \item $q_0^{\R}$ matches $q_0^{\S}$,
      \item States can only be promoted using \emph{prioritized promotion}.
    \end{enumerate}
    We also use the following general assumptions from the paper:
    \begin{enumerate}[label=(\arabic*)]
      \setcounter{enumi}{2}
      \item $P^{\R}$ is minimal,
      \item $P^{\R}$ is prefix-closed,
      \item $\R$ and $\S$ are complete.
    \end{enumerate}

    First, we note that the state cover and separating family are computed on $\R|_{I^{\S}}$, which means that both only contain sequences in the alphabet $I^{\R} \cap I^{\S}$. Because of (3), we know there are $|P^{\R}|$ equivalence classes in the reachable part of $\R|_{I^{\S}}$. From (1) and (5), we derive that for all $w \in (I^{\R} \cap I^{\S})^*$, $\lambda^{\R}(w) = \lambda^{\S}(w)$. This implies that $|P^{\R}| = n$.

    If $|\basis|{}={}n$, we are done. Otherwise, $|\basis|{}< n$. From (1), (3) and $|\basis|{}< n$, we know that some state in $Q^{\S}$, reachable with a sequence from $P^{\R}$, has not been discovered yet. Because of (4), this state must be reachable from the basis with one input symbol. In other words, there must exist a basis state $\firstbasisstate \in \basis$ and $i \in I$ such that $\accessT(\firstbasisstate)i \in P^{\R}$ and $\delta^{\Obs}(\firstbasisstate,i){\uparrow}$. 
      
    From (1), we know that for $\sigma = \mathsf{sep}\bm{(}\delta^{\R}(\accessT(\firstbasisstate)i),\delta^{\R}(\accessT(\secondbasisstate))\bm{)}$ it must hold that $\sigma \vdash \delta^{\S}(\accessT(\firstbasisstate)i) \apart \delta^{\S}(\accessT(\secondbasisstate))$ because $P^{\R}$ and $\{W_q\}^{\R}$ are computed on $\R|_{I^{\S}}$. From (2), we know that for each $\secondbasisstate \in \basis$, $\accessT(\secondbasisstate) \in P^{\R}$.
    
    Therefore, we can apply Lemma~\ref{lem:rebuild} and this will lead to $\delta^{\Obs}(\firstbasisstate,i)$ being \emph{isolated}. Using the \emph{prioritized promotion} rule, we can add $\delta^{\Obs}(\firstbasisstate,i)$ to the basis, leading to $|\basis|{}= n+1$ and we can apply the recursive reasoning to find a new state to promote or to terminate with the required result.

    Note that the precise ordering of \emph{prioritized promotion} and \emph{rebuilding} is irrelevant. We can never promote states that we do not want to promote. Moreover, when a state is isolated, it can never be \textit{un}-isolated. Therefore, applying the \emph{rebuilding} rule while \emph{prioritized promotion} can be applied never leads to problems. Finally, \emph{rebuilding} cannot be applied for ever (see termination proof \ref{thm:quantitativeComplexity}), therefore we have to use \emph{prioritized promotion} at some point.

\end{proof}

\newpage
\section{Proofs of Section~\ref{sec:statematch}}
\subsection*{Proof of Lemma~\ref{lem:statematch}}
\begin{proof}
  Let $\firstrefbasisstate \in Q^{\R}, \firstbasisstate \in \basis, i \in I$ and $\sigma \in I^*$. Suppose 
  \begin{enumerate}[label=(\arabic*)]
    \item $\delta^{\Obs}(\firstbasisstate,i) = \firstfrontierstate \in \frontier$,
    \item $\delta^{\S}(\accessT(\firstbasisstate)) \matches \firstrefbasisstate$,
    \item For all $\secondbasisstate \in \basis$, $\delta^{\R}(p,i) \notmatches \secondbasisstate$.
  \end{enumerate}

  We prove that for all $\secondbasisstate \in \basis$, $\firstfrontierstate \apart \secondbasisstate$ holds. Suppose we have a specific $\secondbasisstate \in \basis$. If $\firstfrontierstate \apart \secondbasisstate$ already holds, we are done. From now, assume (4) $\neg(\firstfrontierstate \apart \secondbasisstate)$. Normally, $\matches$ is not a transitive relation, however, because $\Obs$ is an observation tree for $\S$, $\delta^{\Obs}(\firstbasisstate,w) = \delta^{\S}(\accessT(\firstbasisstate),w)$ for all $w \in I^*$. Therefore, we can derive $\delta^{\Obs}(\firstbasisstate) \matches \firstrefbasisstate$ from (2).
  From (1)-(4), we know all preconditions required for \emph{match separation} hold. We apply the rule and execute OQ $\accessT(\firstbasisstate)i\sigma$ with $\sigma \vdash \secondbasisstate \apart \delta^{\R}(p,i)$. Note here that $\sigma$ with $\sigma \vdash \secondbasisstate \apart \delta^{\R}(p,i)$ must exist due to (3). After the OQ, we have (5) $\delta^{\Obs}(\firstbasisstate,i\sigma){\converges}$ and from (3) we derive (6) $\delta^{\Obs}(\secondbasisstate,\sigma){\converges}$. Because $\sigma \vdash q' \apart \delta^{\R}(p,i)$ and $\delta^{\S}(\accessT(q)) \matches p$, it must be that (7) $\sigma \vdash q' \apart \delta^{\S}(\accessT(q),i)$. Combining (5), (6) and (7) proves $r \apart \secondbasisstate$.

\end{proof}

\subsection*{Proof of Theorem~\ref{thm:statematching}}
\begin{proof}
  Let $\S$ be the SUL and $\R$ the reference with $\S$ and $\R$ both complete Mealy machines. Moreover, let $\S$ be equivalent to $\R$ but $\S$ possibly has a different initial state. From this, we derive that \textbf{(1)} there exists a state $p \in Q^{\R}$ such that $q_0^{\S}$ is language equivalent to $p$.
  Let $n$ be the number of states in the reachable part of $\S$. Let $\basis, \frontier, \Obs$ denote the current basis, frontier and observation tree. 

  We prove that if $|\basis|{}< n$, then we can add some state to the basis after applying \emph{match refinement}, \emph{match separation}, \emph{promotion}, \emph{extension} until none of them are applicable anymore. This trivially terminates when we reach $|\basis|{}= n$. 

  Suppose $|\basis|{}< n$. There must be some $q \in \basis$ and $i \in I$ such that \textbf{(2)} $\delta^{\S}(\accessT(q),i)$ represents an equivalence class that is different from the equivalence classes $\delta^{\S}(\accessT(q'))$ for all $q' \in \basis$. We perform a case distinction on the location of $\delta^{\Obs}(q,i)$ in the current observation tree.
  \begin{itemize}
    \item Suppose $\delta^{\Obs}(q,i) \in \basis$, this immediately contradicts (2).
    \item Suppose $\delta^{\Obs}(q,i){\uparrow}$, then we can apply the \emph{extension} rule resulting in $\delta^{\Obs}(q,i){\converges}$.
    \item Suppose $\delta^{\Obs}(q,i){\converges}$ and $\delta^{\Obs}(q,i)$ is isolated, then we can apply \emph{promotion}.
    \item Suppose \textbf{(3)} $\delta^{\Obs}(q,i){\converges}$ and $\delta^{\Obs}(q,i)$ is not isolated. Moreover, from (1) we derive that there exists a state $p' \in Q^{\R}$ such that \textbf{(4)} $\delta^{\R}(p,\accessT(q))=p'$ and \textbf{(5)} $p'$ is language equivalent to $\delta^{\S}(\accessT(q))$. We perform a case distinction based on whether $\delta^{\R}(p',i) \matches q'$ for some $q' \in \basis$ and show that for each case we can derive a contradiction or apply a rule to make progress.
      \begin{itemize}
        \item 
        Suppose \textbf{(6)} there exists a $q' \in \basis$ such that $\delta^{\R}(p',i) \matches q'$. 
        From (1) we derive that \textbf{(7)} there must exist some state $p'' \in Q^{\R}$ that is language equivalent to $\delta^{\S}(\accessT(q'))$. 

        We derive that state \textbf{(8)} $\delta^{\R}(p',i)$ is language equivalent to an equivalence class that is different from the equivalence classes $\delta^{\S}(\accessT(q'))$ for all $q' \in \basis$ because $\delta^{\R}(p',i)$ is language equivalent to $\delta^{\S}(\accessT(q),i)$ (derived from (5)) and (2).
        Moreover, \textbf{(9)} $p''$ is language equivalent to a state already in the basis (7). 
        By combining (8) and (9), we find that $\delta^{\R}(p',i) \apart p''$.
        Moreover, because $\delta^{\R}(p',i)$ and $p''$ are in $Q^{\R}$ and they represent different equivalence classes, sequence $\sigma = \mathsf{sep}(\delta^{\R}(p',i),p'')$ exists. 
        This means we can apply \emph{match refinement} with $\delta^{\R}(p',i)$ and $p''$, resulting in $\delta^{\R}(p',i) \notmatches q'$ because otherwise (7) leads to a contradiction.\\
        This reasoning can be applied for any $q' \in \basis$ such that $\delta^{\R}(p',i) \matches q'$, resulting in $\delta^{\R}(p',i) \notmatches q'$ for all $q' \in \basis$ after multiple applications of \emph{match refinement}. 
        In this case, we can continue with the case below.

        \item Suppose there does not exist a $q' \in \basis$ such that $\delta^{\R}(p',i) \matches q'$. In this case, we can apply Lemma~\ref{lem:statematch} with $p', q, i$. This results in $\delta^{\Obs}(q,i)$ being isolated. We can apply \emph{promotion} which increases the size of the basis.
      \end{itemize}
      \item Suppose $\delta^{\Obs}(q,i) \notin \basis$ and $\delta^{\Obs}(q,i) \notin \frontier$, this contradicts the assumption that $q \in \basis$ and $i \in I$.
    \end{itemize}

    Note that the precise ordering of the \emph{promotion}, \emph{extension}, \emph{match refinement} and \emph{match separation} is irrelevant. We discuss the reasoning for each rule.
    \begin{description}
      \item[Promotion] States that are isolated can never become \textit{un}-isolated, therefore, applying other rules before \emph{promotion} can never lead to problems.
      \item[Extension] If we apply rules before applying \emph{extension} then either \emph{extension} is not necessary anymore or we can still apply it but both lead to $\delta^{\Obs}(q,i){\converges}$. 
      \item[Match refinement] The only goal of \emph{match refinement} is to refine the matching. If two reference states match a basis state, we can perform an OQ that leads to one of the reference states no longer being a match. If we apply one of the other rules before \emph{match refinement} which already leads to this result, then we do not have to perform \emph{match refinement} but obtain the same result. 
      \item[Match separation] In this Theorem, the \emph{match separation} rule always leads to a new apartness pair between the frontier state and the basis. If some other rule already shows the required apartness pair, we do not have to apply \emph{match separation} but obtain the same result.
    \end{description}
\end{proof}

\subsection*{Proof of Lemma \ref{lem:approx1}}
\begin{proof}
  Let $\Obs$ be an observation tree, $\R$ a reference model, $\firstbasisstate \in Q^{\Obs}$ and $\firstrefbasisstate \in Q^{\R}$. Suppose $\mathsf{mdeg}(\firstbasisstate,\firstrefbasisstate) = 1$. In particular, for all $w \in (I^{\Obs} \cap I^{\R})^*$ and $i \in I^{R} \cap I^{\Obs}$ such that $\delta^{\Obs}(\firstbasisstate,wi){\converges}$,
  \[ \lambda^{\Obs}(\delta^{\Obs}(\firstbasisstate,w),i) = \lambda^{\R}(\delta^{\R}(\firstrefbasisstate,w),i).\] 
  This is equivalent to for all $v \in (I^{\Obs} \cap I^{\R})^*$ such that $\delta^{\Obs}(\firstbasisstate,v){\converges}$
  \[ \lambda^{\Obs}(\firstbasisstate,v) = \lambda^{\R}(\firstrefbasisstate,v) \] 
  Because we assume reference models are complete w.r.t. their own alphabet $I^{R}$, the reference is complete w.r.t. $I^{R} \cap I^{\Obs}$, this implies
  for all $v \in (I^{\Obs} \cap I^{\R})^*$ such that $\delta^{\Obs}(\firstbasisstate,v){\converges}$ and $\delta^{\R}(\firstrefbasisstate,v){\converges}$
  \[ \lambda^{\Obs}(\firstbasisstate,v) = \lambda^{\R}(\firstrefbasisstate,v) \]
  which is precisely $\firstbasisstate \matches \firstrefbasisstate$.
\end{proof}

\section{Proofs of Section~\ref{sec:alsharp}}
Before we prove the complexity for \adaptivelsharp, we define and prove an additional termination Theorem. We prove termination of \adaptivelsharp{} by proving that each rule lowers the ranking function. To keep consistent with the \lsharp{} complexity proof~\cite{DBLP:conf/tacas/VaandragerGRW22}, we actually prove that each rule increases some norm and the norm is bounded by the SUL. Specifically, we use norm $N(\Obs)$:
\begin{equation} \label{eq:normQuant}
    N(\Obs) = N_{L^{\#}}(\Obs) + |N_{(\basis\times Q^{\R} \times Q^{\R})\converges}(\Obs)|~+~|N_{\frontier \apart Q^{\R}}(\Obs)|~+~|N_{(\basis\times \frontier)\converges}(\Obs)| 
\end{equation}
where $N_{L^{\#}}(\Obs)$ indicates the slightly adapted norm from \cite{DBLP:conf/tacas/VaandragerGRW22}. The abbreviations for the summands are defined as follows.
\begin{align*}
&N_{L^{\#}}(\Obs) = |\basis|(|\basis|+1) + |\{ (\firstbasisstate, i) \in \basis \times I \mid \delta(\firstbasisstate,i) \converges \}| + |\{ (\firstbasisstate, \firstfrontierstate) \in \basis \times \frontier \mid \firstbasisstate \apart \firstfrontierstate \}| \\
&N_{(\basis\times Q^{\R} \times Q^{\R})\converges}(\Obs) = \{(\firstbasisstate,\firstrefbasisstate,\secondrefbasisstate) \in \basis \times Q^{\R} \times Q^{\R} \mid \delta(\firstbasisstate,\sigma)\converges \text{ with } \sigma = \mathsf{sep}(\firstrefbasisstate,\secondrefbasisstate) \} \\
&N_{(\basis \cup \frontier) \apart Q^{\R}}(\Obs) = \{(\firstbasisstate,\firstrefbasisstate) \in (\basis \cup \frontier) \times Q^{\R} \mid \firstbasisstate \apart \firstrefbasisstate \} \\
&N_{(\basis\times \frontier)\converges}(\Obs) = \{(\firstbasisstate,\firstfrontierstate) \in \basis \times \frontier \mid~\delta(\firstbasisstate,\sigma)\converges \land \delta(\firstfrontierstate,\sigma)\converges \text{ with } \\ 
    &\qquad \qquad \qquad \qquad \qquad \sigma = \mathsf{sep}\bm{(}\delta^{\R}(\accessT(\firstbasisstate)),\delta^{\R}(\accessT(\firstfrontierstate))\bm{)}\}
\end{align*}
The summand $N_{(\basis \times Q^{\R} \times Q^{\R})\converges}(\Obs)$ keeps track of which separating sequences of the reference model have been applied to basis states in the new observation tree. The summand $N_{(\basis \cup \frontier) \apart Q^{\R}}(\Obs)$ keeps track of unmatched states between states in the basis or frontier and the reference model. The summand $N_{(\basis \times \frontier)\converges}(\Obs)$ keeps track of separating sequences from the reference model applied to pairs of basis and frontier states. These summands are motivated by the postconditions in Table~\ref{tab:rules}.

\begin{theorem} \label{lem:progress}
    Every rule application in \adaptivelsharp{} increases norm $N(\Obs)$.
\end{theorem}
\begin{proof}
  Let $\basis,\frontier,\Obs$ denote the values before and $\basis',\frontier',\Obs'$ denote
  the values after the respective rule application. Let $\R$ denote the reference model. We reuse abbreviations from \cite{DBLP:conf/tacas/VaandragerGRW22} and the norm definition above
  \begin{align*}
    &N_Q(\Obs) = |\basis| \cdot (|\basis|+1)\\
    &N_{\converges}(\Obs) = \{ (\firstbasisstate, i) \in \basis \times I \mid \delta(\firstbasisstate,i) \converges \}\\
    &N_{\apart}(\Obs) = \{ (\firstbasisstate, \firstfrontierstate) \in \basis \times \frontier \mid \firstbasisstate \apart \firstfrontierstate \}
  \end{align*}

The proof that the rules \emph{promotion}, \emph{extension}, \emph{separation} and \emph{equivalence} increase $N_Q(\Obs) + |N_{\converges}(\Obs)| + |N_{\apart}(\Obs)|$ is similar to the proof in~\cite{DBLP:conf/tacas/VaandragerGRW22}. However, we slightly adapted $N_Q(\Obs)$ which has influence on the proof for \emph{promotion} and \emph{separation} and we assume stronger guarantees for \emph{equivalence}. It remains to show that combined with the new summands the total norm still increases. Therefore, we include the proofs for the \lsharp{} rules here.

\begin{description}
  \item[Rebuilding] Let $\firstbasisstate, \secondbasisstate \in \basis$, $i \in I$ and $\sigma \in I^*$. We assume 
  \begin{enumerate}
    \item $\delta^{\Obs}(\firstbasisstate,i) \notin \basis$
    \item $\neg(\secondbasisstate \apart \delta^{\Obs}(\firstbasisstate,i))$,
    \item $\accessT(\firstbasisstate)$, $\accessT(\firstbasisstate)i \in P^{\R}$,
    \item $\delta^{\Obs}(\firstbasisstate,i\sigma){\uparrow}$ or $\delta^{\Obs}(\secondbasisstate,\sigma){\uparrow}$
    \item $\sigma = \mathsf{sep}\bm{(}\delta^{\R}(\accessT(\firstbasisstate)i),\delta^{\R}(\accessT(\secondbasisstate))\bm{)}$.
  \end{enumerate}
  
  The algorithm performs two queries OQs $\accessT(\firstbasisstate)i\sigma$ and $\accessT(\secondbasisstate)\sigma$.
  After these OQs, the traces $\delta^{\Obs}(\firstbasisstate,i\sigma)$ and $\delta^{\Obs}(\secondbasisstate,\sigma)$ are defined. Particularly, because of assumption (1) and $\firstbasisstate \in \basis$, we know $\delta^{\Obs}(\firstbasisstate,i) \in \frontier$ after the OQs. Combining this with (4), we find 
  \[ N_{(\basis\times \frontier)\converges}(\Obs') \supseteq N_{(\basis\times\frontier)\converges}(\Obs) \cup \{(\secondbasisstate,\delta^{\Obs}(\firstbasisstate,i))\} \]
  Note that we implicitly use (3) and (5) to ensure that $\sigma$ exists. In some cases we might find that $\delta^{\Obs}(\firstbasisstate,i) \apart \secondbasisstate$ which, together with (2), indicates 
  \[ N_{\apart}(\Obs') \supseteq N_{\apart}(\Obs) \cup \{(\secondbasisstate,\delta^{\Obs}(\firstbasisstate,i))\} \] Otherwise $N_{\apart}(\Obs') \supseteq N_{\apart}(\Obs)$. Additionally,
  \[ N_Q(\Obs') = N_Q(\Obs) \]
  \[ N_\downarrow(\Obs') \supseteq N_\downarrow(\Obs) \] 
  \[ N_{(\basis \cup \frontier) \apart Q^{\R}}(\Obs') \supseteq  N_{(\basis \cup \frontier) \apart Q^{\R}}(\Obs) \]
  \[ N_{(\basis\times Q^{\R} \times Q^{\R})\converges}(\Obs') \supseteq N_{(\basis\times Q^{\R} \times Q^{\R})\converges}(\Obs) \]
  Thus, $N(\Obs') \ge N(\Obs) + 1$.

  \item[Prioritized promotion] Let $\firstfrontierstate \in \frontier$. Suppose (1) $\firstfrontierstate$ is isolated and suppose (2) $\accessT(\firstfrontierstate) \in P^{\R}$. State $\firstfrontierstate$ is moved from $\frontier$ to $\basis$, i.e.~$\basis' := \basis \cup \{\firstfrontierstate\}$, then we have 
  \begin{align*}
    N_Q(\Obs') &= |\basis'| \cdot (|\basis'|+1)
    = (|\basis|+1) \cdot (|\basis|+1+1)\\
               &= (|\basis|+1)\cdot |\basis| + 2(|\basis|+1)
                 = N_Q(\Obs) + 2|\basis| + 2
  \end{align*}
  
  Because we move something from the frontier to the basis, we find
  \[ N_{\apart}(\Obs') ~~\supseteq~~ N_{\apart}(\Obs) \setminus (\basis\times \{\firstfrontierstate\}) \] 
  \[ N_{(\basis\times\frontier)\converges}(\Obs) ~~\supseteq~~ N_{(\basis\times\frontier)\converges}(\Obs) \setminus (\basis\times \{\firstfrontierstate\})
  \]
  and thus
  \[ |N_{\apart}(\Obs')| ~~\ge~~ |N_{\apart}(\Obs)| - |\basis| \] 
  \[ |N_{(\basis\times \frontier)\converges}(\Obs')| ~~\ge~~ |N_{(\basis\times\frontier)\converges}(\Obs)| - |\basis|
  \]
  Finally, 
  \[ N_\downarrow(\Obs') \supseteq N_\downarrow(\Obs) \]
  \[ N_{\basis \times Q^{\R} \times Q^{\R}}(\Obs') \supseteq N_{\basis \times Q^{\R} \times Q^{\R}}(\Obs) \]
  \[ N_{(\basis \cup \frontier) \apart Q^{\R}}(\Obs') \supseteq N_{(\basis \cup \frontier) \apart Q^{\R}}(\Obs)\]
  The total norm increases because
   \[ N(\Obs') \ge N(\Obs)+ 2\mid \basis\mid +~2~- \mid \basis\mid - \mid \basis\mid ~\ge N(\Obs)+ 2 \]  
  \item[Promotion] Analogous to the proof for prioritized promotion.
  \item[Extension] Let $\delta^{\Obs}(\firstbasisstate,i){\uparrow}$ for some $\firstbasisstate \in \basis$, $i\in I$. After OQ $\accessT(\firstbasisstate)i$, we get $N(\Obs') \ge N(\Obs) + 1$ from~\cite{DBLP:conf/tacas/VaandragerGRW22}. Additionally,
  \[
    N_{(\basis\times Q^{\R} \times Q^{\R})\converges}(\Obs') \supseteq N_{(\basis\times Q^{\R} \times Q^{\R})\converges}(\Obs) \] 
    \[
        N_{(\basis\times \frontier)\converges}(\Obs') \supseteq N_{(\basis\times\frontier)\converges}(\Obs) \] 
    \[
        N_{(\basis \cup \frontier) \apart Q^{\R}}(\Obs') \supseteq N_{(\basis \cup \frontier) \apart Q^{\R}}(\Obs)
    \]
    and thus $N(\Obs') \ge N(\Obs) + 1$.
  \item[Separation] Consider a state $\firstfrontierstate \in \frontier$ and distinct $\firstbasisstate, \secondbasisstate \in \basis$
  with $\neg(\firstfrontierstate \apart \firstbasisstate)$ and $\neg(\firstfrontierstate \apart \secondbasisstate)$. After OQ
  $\accessT(\firstbasisstate)\sigma$, we have $N(\Obs') \ge N(\Obs) + 1$ from~\cite{DBLP:conf/tacas/VaandragerGRW22}. Additionally,
    \[ N_{(\basis\times Q^{\R} \times Q^{\R})\converges}(\Obs') \supseteq N_{(\basis\times Q^{\R} \times Q^{\R})\converges}(\Obs) \] 
    \[ N_{(\basis\times \frontier)\converges}(\Obs') \supseteq N_{(\basis\times\frontier)\converges}(\Obs) \]
    \[ N_{(\basis \cup \frontier) \apart Q^{\R}}(\Obs') \supseteq N_{(\basis \cup \frontier) \apart Q^{\R}}(\Obs) \]
  Thus, $N(\Obs') \ge N(\Obs) + 1$.
  \item[Match separation] Let $\firstbasisstate \in \basis$, $\firstrefbasisstate \in Q^{\R}, i \in I$, $\sigma \in I^*$, $\delta^{\Obs}(\firstbasisstate,i)=\firstfrontierstate \in \frontier$ and $\delta^{\R}(\firstrefbasisstate,i)=\secondrefbasisstate$. Suppose 
  \begin{enumerate}
    \item $\firstbasisstate \approxmatches \firstrefbasisstate$,
    \item $\neg(\firstfrontierstate \apart \secondrefbasisstate)$, 
    \item There is no $\thirdbasisstate \in \basis$ such that $\secondrefbasisstate \approxmatches \thirdbasisstate$,
    \item There exists $\secondbasisstate \in \basis$ such that 
    $\neg(\firstfrontierstate \apart \secondbasisstate)$ and $\sigma \vdash \secondrefbasisstate \apart \secondbasisstate$.
  \end{enumerate}
  After OQ $\accessT(\firstbasisstate)i\sigma$ we find either
  
  \[ \firstfrontierstate \apart \secondbasisstate \quad \text{or} \quad \firstfrontierstate \apart \secondrefbasisstate \]

  If $\firstfrontierstate \apart \secondbasisstate$, then 
  \[ N_{\apart}(\Obs') \supseteq N_{\apart}(\Obs) \cup \{(\secondbasisstate,\firstfrontierstate)\} \qquad N_{(\basis \cup \frontier) \apart Q^{\R}}(\Obs') \supseteq N_{(\basis \cup \frontier) \apart Q^{\R}}(\Obs)\]
  
  If $\firstfrontierstate \apart \secondrefbasisstate$, then 
  \[ N_{\apart}(\Obs') \supseteq N_{\apart}(\Obs) \qquad N_{(\basis \cup \frontier) \apart Q^{\R}}(\Obs') \supseteq N_{(\basis \cup \frontier) \apart Q^{\R}}(\Obs) \cup \{(\firstfrontierstate,\secondrefbasisstate)\} \]
  Additionally, 
  \[ N_Q(\Obs') = N_Q(\Obs) \]
  \[ N_\downarrow(\Obs') \supseteq N_{\downarrow}(\Obs) \]
  \[ N_{(\basis\times \frontier)\converges}(\Obs') \supseteq N_{(\basis\times\frontier)\converges}(\Obs) \]
  \[ N_{\basis \times Q^{\R} \times Q^{\R}}(\Obs') \supseteq N_{\basis \times Q^{\R} \times Q^{\R}}(\Obs) \]
  Thus, $N(\Obs') \ge N(\Obs) + 1$.
  \item[Match refinement] Let $\firstbasisstate \in \basis$ and $\firstrefbasisstate, \secondrefbasisstate \in Q^{\R}$. Suppose $\firstbasisstate \approxmatches \firstrefbasisstate$ and $\firstbasisstate \approxmatches \secondrefbasisstate$. 
  Note that when using approximate matching this \textbf{does not} imply that $\neg(\firstbasisstate \apart \firstrefbasisstate)$ and $\neg(\firstbasisstate \apart \secondrefbasisstate)$. After OQ $\accessT(\firstbasisstate)\sigma$ with $\sigma = \mathsf{sep}(p,p')$, we find
  \[
  N_{\basis \times Q^{\R} \times Q^{\R}}(\Obs') \supseteq N_{\basis \times Q^{\R} \times Q^{\R}}(\Obs) \cup \{(\firstbasisstate,\firstrefbasisstate,\secondrefbasisstate)\}
  \]
  Additionally, 
  \[ N_\downarrow(\Obs') \supseteq N_{\downarrow}(\Obs) \]
  \[ N_{\apart}(\Obs') \supseteq N_{\apart}(\Obs) \]
  \[ N_{(\basis\times \frontier)\converges}(\Obs') \supseteq N_{(\basis\times\frontier)\converges}(\Obs) \]
  \[ N_{(\basis \cup \frontier) \apart Q^{\R}}(\Obs') \supseteq N_{(\basis \cup \frontier) \apart Q^{\R}}(\Obs) \]
  Because $N_{Q}$ remains unchanged, we have $N(\Obs') \ge N(\Obs) + 1$.
  \item[Prioritized separation] Analogous to the proof for separation, the additional condition on $\sigma$ does not change the postcondition.

  \item[Equivalence] Suppose all $\firstfrontierstate \in \frontier$ are identified and for all $\firstbasisstate \in \basis$ and $i \in I$, $\delta^{\Obs}(\firstbasisstate,i){\downarrow}$. These conditions are stronger than the conditions from~\cite{DBLP:conf/tacas/VaandragerGRW22}. Therefore, we know at least $N_{L^{\#}}(\Obs') \ge N_{L^{\#}}(\Obs) + 1$ holds. Additionally,
  \[ N_{(\basis\times Q^{\R} \times Q^{\R})\converges}(\Obs') \supseteq N_{(\basis\times Q^{\R} \times Q^{\R})\converges}(\Obs) \] 
  \[ N_{(\basis\times \frontier)\converges}(\Obs') \supseteq N_{(\basis\times\frontier)\converges}(\Obs) \]
    \[ N_{(\basis \cup \frontier) \apart Q^{\R}}(\Obs') \supseteq N_{(\basis \cup \frontier) \apart Q^{\R}}(\Obs) \]
  Thus, $N(\Obs') \ge N(\Obs) + 1$.
\end{description}
\end{proof}

\subsection*{Proof of Theorem \ref{thm:quantitativeComplexity}}
\begin{proof}
  First, we prove that if $\Obs$ is an observation tree for $\S$, then 
  \begin{align*}
    N(\Obs) &\leq n(n+1) + kn + (n-1)(kn+1) + n\nrrefstates^2 + (kn+1)\nrrefstates + n(kn+1) \\
    &\in \bigO(kn^2 + kn\nrrefstates + n\nrrefstates^2) 
  \end{align*}
  The first part $n(n+1) + kn + (n-1)(kn+1)$ follows from Theorem 3.9 in~\cite{DBLP:conf/tacas/VaandragerGRW22} with some minor adjustments. 
  The set $\basis$ contains at most $n$ elements and $Q^{\R}$ contains at most $\nrrefstates$ elements. Each state in $Q^{\R}$ can be apart from at most $|Q^{\R}|-1$ other states in $Q^{\R}$. Therefore,
  \[ |\{(q,p,p') \in \basis \times Q^{\R} \times Q^{\R} \mid \delta(q,\sigma)\converges \text{ with } \sigma = \mathsf{sep}(p,p') \}| {}\leq no(\nrrefstates-1) \leq no^2 \]
  
  Since the set $\basis \cup \frontier$ contains at most $kn+1$ elements and each state in $\basis \cup \frontier$ can be apart from at most $\nrrefstates$ states from $Q^{\R}$, we have 
  \begin{eqnarray*}
      | \{ (q, q') \in (\basis \cup \frontier) \times Q^{\R} \mid q \apart q' \} | & \leq & \nrrefstates(kn +1)
    \end{eqnarray*}

  The set $\frontier$ contains at most $kn$ elements and each pair $\basis \times \frontier$ has at most one $\sigma = \mathsf{sep}(\delta^{\R}(\accessT(q)),\delta^{\R}(\accessT(r))$, thus we have 
  \[ |N_{(\basis\times\frontier)\converges}| {}\leq kn^2 \]

  Combining everything and simplifying it leads to
  \[N(\Obs) \in \bigO(kn^2 + kn\nrrefstates + n\nrrefstates^2) \]

  The ordering on the rules never block the algorithm and when the norm $N(\Obs)$ cannot be increased further, the only applicable rule is the equivalence rule which is guaranteed to lead to the teacher accepting the hypothesis. Therefore, the correct Mealy machine is learned within $\bigO(kn^2 + kn\nrrefstates + n\nrrefstates^2)$ rule applications. 

  In \adaptivelsharp, every (non-terminating) application of the equivalence rule leads to a new basis state. Since the basis is bounded by the number of states in the SUL, which is $n$, there can be at most $n-1$ applications of the equivalence rule. Each call to \textsc{ProcCounterEx} requires at most $\log m$ output queries (see Theorem 3.11 of~\cite{DBLP:conf/tacas/VaandragerGRW22}). 

  All rules, except for the equivalence rule, require at most two OQs per rule application. Therefore, the application of these rules requires $\bigO(kn^2 + kn\nrrefstates + n\nrrefstates^2)$ OQs. Combining everything, we find that \adaptivelsharp requires $\bigO(kn^2 + kn\nrrefstates + n\nrrefstates^2 + n \log m)$ and at most $n-1$ EQs. 
\end{proof}

\section{Additional Experiment Information} \label{app:E}
\subsection{Experiment Models}
In Experiments 1 and 3, we use the following six models, available \href{https://automata.cs.ru.nl/Overview}{here} under \textit{Mealy machine benchmarks}.
\begin{itemize}
  \item \textit{learnresult\_fix}
  \item \textit{DropBear} 
  \item \textit{OpenSSH} 
  \item \textit{model1} 
  \item \textit{NSS\_3.17.4\_server\_regular} 
  \item \textit{GnuTLS\_3.3.8\_client\_full}
\end{itemize}
Due to the mutations, this means that the largest model that we can learn has 62 states ($\textit{mut}_{8}(\S)$).
In Experiment 2, we use the ordering for the \textit{Adaptive-OpenSSL} models as implied by Fig.~5 in \cite{DBLP:conf/ifm/DamascenoMS19}. The ordering taken for the \textit{Adaptive-Philips} is chronological.
In Experiment 4, we use the \href{https://automata.cs.ru.nl/BenchmarkTCP/Mealy}{client TCP models}. Additionally, we use the following \href{https://automata.cs.ru.nl/BenchmarkDTLS-Fiterau-BrosteanEtAl2020/Mealy}{DTLS models}.
\begin{itemize}
  \item \textit{ctinydtls\_ecdhe\_cert\_req.dot} 
  \item \textit{etinydtls\_ecdhe\_cert\_req.dot} 
  \item \textit{gnutls-3.6.7\_all\_cert\_req.dot} 
  \item \textit{mbedtls\_all\_cert\_req.dot} 
  \item \textit{scandium-2.0.0\_ecdhe\_cert\_req.dot} 
  \item \textit{scandium\_latest\_ecdhe\_cert\_req.dot} 
  \item \textit{wolfssl-4.0.0\_dhe\_ecdhe\_rsa\_cert\_req.dot}
\end{itemize}

\subsection{Mutation Explanations}
In this section, we call the input Mealy machine $\S = (Q,I,O,q_0,\delta,\lambda)$. Every mutation is applied exactly once to generate the mutated model.

\myparagraph{$\textit{mut}_{1}$: New initial state.} This mutation adds a new initial state called \emph{dummy} as well as a fresh symbol $i$ to $\S$. From state \emph{dummy}, all $i \in I$ self loop with the output from $q_0$ (the previous initial state). The fresh symbol transitions from \emph{dummy} to $q_0$. The fresh symbol self loops in all (other) states $q \in Q$ with the output of $\lambda(q_0,i_0)$ where $i_0$ is the first input in the alphabet.

\myparagraph{$\textit{mut}_{2}$: Change the initial state.} This mutation randomly selects one of the states in $Q$ to pick as the new initial state. Because $\S$ is not necessarily strongly connected, the number of states in the resulting Mealy machine might be lower.

\myparagraph{$\textit{mut}_{3}$: Add a state.} This mutation adds a new state $q$ to $Q$. We randomly select a state from $q' \in Q$ and $i \in I$ and change the destination of this transition to $q$, this ensures $q$ is reachable. For all $i \in I$, we randomly select a destination state $p$ and use the output $\lambda(p,i)$ $80\%$ of the time or a random output $20\%$.

\myparagraph{$\textit{mut}_{4}$: Remove a state.} This mutation removes a non-initial state $q$ from $\S$. All transitions that lead from $p$ to $q$ with input $i$ are shortcutted to $\delta(q,i)$ with output $\lambda(p,i)$. If $\delta(q,i)=q$, we self loop in $p$.

\myparagraph{$\textit{mut}_{5}$: Divert a transition.} This mutation randomly selects $q,q' \in Q$ and $i \in I$. We set $\delta(q,i)=q'$. While $\S$ is equivalent to the resulting Mealy machine, we choose a new $q, q', i$ and set $\delta(q,i)=q'$.

\myparagraph{$\textit{mut}_{6}$: Change transition output.} This mutation randomly selects $q \in Q$, $i \in I$ and $o \in O$. We set $\lambda(q,i)=o$ such that $o$ is distinct from the original $\lambda(q,i)$. 

\myparagraph{$\textit{mut}_{7}$: Remove a symbol.} This mutation removes a symbol $i$ from the input alphabet. Consequently, all the transitions with $i$ are not contained in the resulting Mealy machine.

\myparagraph{$\textit{mut}_{8}$: Appending a mutated model.} This mutation takes a $\S$ and a natural number $n$. It first makes a second Mealy machine $\S'$ by applying $\textit{mut}_{13}$ to $\S$. Then it appends $\S'$ to $\S$ at the $n^{\text{th}}$ state of $\S$ which we call $q$, i.e., $\delta(q,i)=q_0^{\S'}$ for some random $i$. The natural numbers are chosen based on visual inspection of the models, we consistently choose a state that represents the \textit{end} of a model. This \emph{end} state is either the sink state or a state at the end of a very long trace in the model which transitions to the sink state.

\myparagraph{$\textit{mut}_{9}$: Prepending a mutated model.} This mutation takes a $\S$ and a natural number $n$. It first makes a second Mealy machine $\S'$ by appling $\textit{mut}_{13}$ to $\S$. Then it appends $\S$ to $\S'$ at the $n^{\text{th}}$ state of $\S'$ which we call $q$, i.e., $\delta(q,i)=q_0^{\S}$ for some random $i$.

\myparagraph{$\textit{mut}_{10}$: Several mutations.} This mutation applies mutations $\textit{mut}_{3}$, $\textit{mut}_{4}$, $\textit{mut}_{5}$ and $\textit{mut}_{6}$ to $\S$ in this particular order. 

\myparagraph{$\textit{mut}_{11}$: Several mutations with different initial state.} This mutation applies $\textit{mut}_{2}$, $\textit{mut}_{3}$, $\textit{mut}_{4}$, $\textit{mut}_{5}$ and $\textit{mut}_{6}$ to $\S$ in this particular order.

\myparagraph{$\textit{mut}_{12}$: Changing many transitions.} This mutation applies $\textit{mut}_{5}$, $\textit{mut}_{6}$, $\textit{mut}_{5}$, $\textit{mut}_{6}$, $\textit{mut}_{5}$, $\textit{mut}_{6}$ to $\S$.

\myparagraph{$\textit{mut}_{13}$: Many mutations.} This mutation applies $\textit{mut}_{10}$ three times to $\S$.

\myparagraph{$\textit{mut}_{14}$: Union.} This mutation takes $\S$ and makes a second Mealy machine $\S'$ by applying $\textit{mut}_{13}$ to $\S$. We combine $\S$ and $\S'$ by creating one new \emph{dummy} initial state with two fresh symbols for which one goes to $q_0^{\S}$ and the other to $q_0^{\S'}$. The fresh symbols and other transitions are handled in the same way as in $\textit{mut}_1$.

\subsection{Additional Figure Experiment 1}
In Fig.~\ref{fig:res1}, we show additional pairwise comparison plots from Experiment 1. Each plot compares a pair of algorithms per model and mutation, where a point $(x,y)$ represents that the algorithm on the x-axis required $x$ symbols over all seeds and the algorithm on the y-axis requires $y$ symbols.  Points below the diagonal indicate that the y-algorithm outperforms the x-algorithm, points below the dashed (dotted) line indicate a factor two (ten) improvement, respectively.
\begin{figure}[ht]
  \begin{subfigure}[b]{0.49\textwidth}
      \centering
      \resizebox{.75\textwidth}{!}{
      \includegraphics{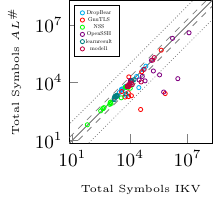}}
      \label{fig:res1a}
  \end{subfigure}
  \begin{subfigure}[b]{0.49\textwidth}
      \centering
      \resizebox{.75\textwidth}{!}{
      \includegraphics{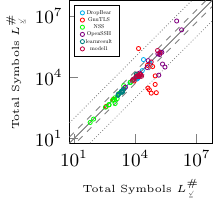}}
      \label{fig:res1b}
  \end{subfigure}

  \begin{subfigure}[b]{0.49\textwidth}
    \centering
    \resizebox{.75\textwidth}{!}{
    \includegraphics{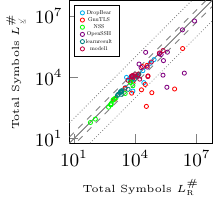}}
    \label{fig:res1c}
\end{subfigure}
\begin{subfigure}[b]{0.49\textwidth}
  \centering
  \resizebox{.75\textwidth}{!}{
  \includegraphics{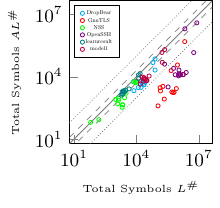}}
  \label{fig:res1d}
\end{subfigure}

\begin{subfigure}[b]{0.49\textwidth}
  \centering
  \resizebox{.75\textwidth}{!}{
  \includegraphics{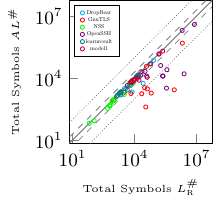}}
  \label{fig:res1e}
\end{subfigure}
\begin{subfigure}[b]{0.49\textwidth}
  \centering
  \resizebox{.75\textwidth}{!}{
  \includegraphics{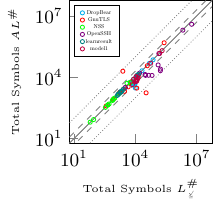}}
\end{subfigure}
\caption{Pairwise comparisons between algorithms.}
  \label{fig:res1}
\end{figure}

\newpage

\subsection{Additional Tables Experiment 2}
Tables \ref{tab:my_exp2_openssl} and \ref{tab:my_exp2_philips} display the mean number of inputs per model and algorithm in the same style as Table~\ref{tab:my_exp1}. The reference row indicates which reference model was used for the adaptive algorithms. THe teal values indicate the lowest, and therefore, best score.
\input{finalversionfigures/exp2_tablePhilips.tex}

\input{finalversionfigures/exp2_tableOpenSSL.tex}

%% file: figures/rules.tex
\begin{figure}[ht!]

  \begin{subfigure}{.475\textwidth}
      \resizebox{.99\textwidth}{!}{
          \includegraphics{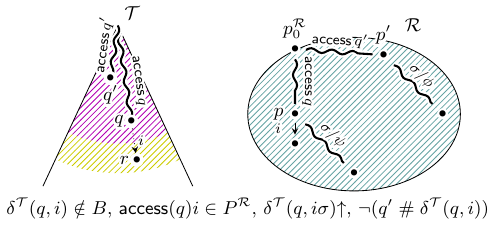}}
    \caption{Rebuilding}
    \label{fig:rebuild}
  \end{subfigure}\hfill 
  \begin{subfigure}{.475\textwidth}
      \resizebox{.99\textwidth}{!}{
        \includegraphics{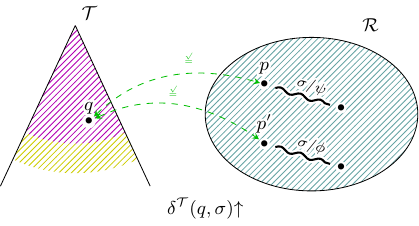}}
        \caption{Match refinement}
        \label{fig:matchrefinement}
  \end{subfigure}
  
  \medskip 
  \begin{subfigure}{.475\textwidth}
      \resizebox{.99\textwidth}{!}{
        \includegraphics{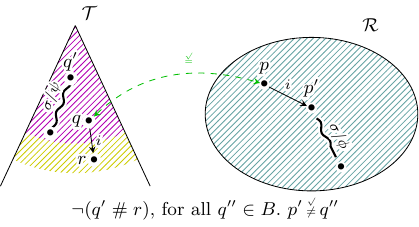}}
        \caption{Match separation}
        \label{fig:matchseparation}
  \end{subfigure}\hfill 
  \begin{subfigure}{.475\textwidth}
      \resizebox{.99\textwidth}{!}{
      \includegraphics{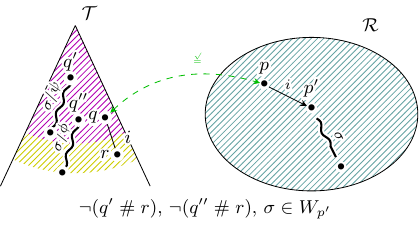}}
    \caption{Prioritized separation}
    \label{fig:prioiritizedseparation}
  \end{subfigure}
  
  \caption{Scenario in the observation tree (left) and reference model (right) required to apply the rule, additional preconditions are written below the scenario. The dashed green lines indicate that two states are matched. }
  \label{fig:rules}
  \end{figure}

%% file: alglisting.tex
\begin{algorithm}[!ht]
    \begin{algorithmic}
      \Procedure{ExtendedLSharp}{$P^{\R}, \{W_q\}^{\R}$}
      \State $\Obs \xleftarrow{} \{q_0\}$ s.t. $\delta^{\Obs}(\varepsilon)=q_0$
      \State $\basis \xleftarrow{} \{q_0\}$

      \DoIf{$\delta^{\Obs}(\firstbasisstate,i) \notin \basis$ and $ \neg(\secondbasisstate \apart \delta^{\Obs}(\firstbasisstate,i)) $ for $\firstbasisstate, \secondbasisstate \in \basis$, $i \in I$ s.t.\\ \qquad \qquad $\mathsf{access}(\firstbasisstate)i, \mathsf{access}(\secondbasisstate) \in P^{\R}$ and ($\delta^{\Obs}(\firstbasisstate,i\sigma){\uparrow}$ or $\delta^{\Obs}(\secondbasisstate,\sigma){\uparrow}$)\\ \qquad \qquad 
      $\sigma = \mathsf{sep}\bm{(}\delta^{\R}(\accessT(\firstbasisstate)i),\delta^{\R}(\accessT(\secondbasisstate))\bm{)}$
      }
        \Comment{\emph{rebuilding}}
                \State \textsc{OutputQuery}($\mathsf{access}(\firstbasisstate)i\sigma$)
                \State \textsc{OutputQuery}($\mathsf{access}(\secondbasisstate)\sigma$) 
    \ElsDoIf{$\firstfrontierstate \in \frontier$ is isolated and $\mathsf{access}(\firstfrontierstate) \in P^{\R}$}
    \Comment{\emph{prefix promotion}}
    \State $\basis \gets \basis \cup \{ \firstfrontierstate \}$
    \ElsDoIf{$\firstfrontierstate \in \frontier$ is isolated}
      \Comment{\emph{promotion}}
      \State $\basis \gets \basis \cup \{ \firstfrontierstate \}$
		\ElsDoIf{$\delta^{\Obs}(\firstbasisstate,i){\uparrow}$, for some $\firstbasisstate \in \basis, i \in I$}
      \Comment{\emph{extension}}
      \State $\OutputQuery(\mathsf{access}(\firstbasisstate)i)$
      \ElsDoIf{$\neg(\firstfrontierstate \apart \firstbasisstate)$, $\neg(\firstfrontierstate \apart \secondbasisstate)$, for some $\firstfrontierstate \in \frontier$, $\firstbasisstate, \secondbasisstate\in
      \basis$, $\firstbasisstate\neq \secondbasisstate$}
      \Comment{\emph{separation}}
      \State $\sigma \gets \text{witness of \(\firstbasisstate \apart \secondbasisstate\)}$
      \State $\OutputQuery(\mathsf{access}(\firstfrontierstate) \sigma)$

      \ElsDoIf{$\firstrefbasisstate \approxmatches \firstbasisstate$ for some $\firstbasisstate \in \basis$ and there is some $i \in I$ s.t.\\ \qquad \qquad $\delta^{\Obs}(\firstbasisstate,i)=\firstfrontierstate \in \frontier$, $\neg(\firstfrontierstate \apart \secondbasisstate)$ for some $\secondbasisstate \in \basis$ and\\ \qquad \qquad $\neg(\firstfrontierstate \apart \secondrefbasisstate)$ for $\delta^{\R}(\firstrefbasisstate,i)=\secondrefbasisstate$ and\\ \qquad \qquad $\secondrefbasisstate \notapproxmatches \thirdbasisstate$ for any $\thirdbasisstate \in \basis$} 
      \Comment{\emph{match separation}}
      \State $\sigma \gets \mbox{witness for } \secondbasisstate \apart \secondrefbasisstate$
      \State $\textsc{OutputQuery}(\mathsf{access}(\firstbasisstate) i \sigma)$
          
    \ElsDoIf{$\firstrefbasisstate \approxmatches \firstbasisstate$ and $\secondrefbasisstate \approxmatches \firstbasisstate$  for some $\firstbasisstate \in \basis$ and $\firstrefbasisstate, \secondrefbasisstate \in Q^{\R}$ with\\ \qquad \qquad $\sigma = \mathsf{sep}(\firstrefbasisstate,\secondrefbasisstate)$ and $\delta^{\Obs}(\firstbasisstate,\sigma){\uparrow}$} \Comment{\emph{match refinement}}
    \State $\textsc{OutputQuery}(\mathsf{access}(\firstbasisstate)\sigma)$

    \ElsDoIf{$\neg(\firstfrontierstate \apart \secondbasisstate)$, $\neg(\firstfrontierstate \apart \thirdbasisstate)$, for some $\firstfrontierstate \in \frontier$, $\firstbasisstate, \secondbasisstate, \thirdbasisstate\in
    \basis$ s.t.\\
    \qquad \qquad $\delta^{\Obs}(\firstbasisstate,i)=\firstfrontierstate$ for some $i \in I$, $\sigma \vdash \firstbasisstate \apart \secondbasisstate$\\ \qquad \qquad $\sigma \in \cup_{\firstrefbasisstate \approxmatches \firstbasisstate} W_{\delta^{\R}(\firstrefbasisstate,i)}$} 
    \Comment{\emph{prioritized separation}}
    \State $\OutputQuery(\mathsf{access}(\firstfrontierstate)\sigma)$

		\ElsDoIf{All $\firstfrontierstate \in \frontier$ are identified and $\delta^{\Obs}(q,i){\converges}$ for all $q \in \basis, i \in I$}
      \Comment{\emph{equivalence}}
		  \State $\Hyp \gets \BuildHypothesis$ 
		  \State $(b, \sigma) \gets \CheckConsistency(\Hyp)$		
		  \If{$b = \code{yes}$}
        \State $(b, \rho) \gets \EquivalenceQuery(\Hyp)$
        \StateIf{$b = \code{yes}$} \Return $\Hyp$
        \StateElse
          $\sigma \gets$ shortest prefix of $\rho$ such that
          $\delta^{\Hyp}(q_0^\Hyp, \sigma) \apart \delta^{\Obs}(q_0^\Obs, \sigma)$
        (in $\Obs$)
		  \EndIf
		  \State $\ProcCounterEx(\Hyp, \sigma)$
    \EndDoIf
        \EndProcedure
    \end{algorithmic}
    \caption{Extended $L^{\#}$ algorithm}
    \label{alg:extendedlsharp}
\end{algorithm}

%% file: finalversionfigures/exp2_tablePhilips.tex
\renewcommand{\arraystretch}{1.4}
\begin{table}[t]
\centering
\caption{Mean inputs for learning a Philips model with a reference.}
\resizebox{.5\textwidth}{!}{
\begin{tabular}{lrrrrrrr}
\toprule
Algorithm & ~model2~ & ~model3~ & ~model4~ & ~model5~ & ~model6~ \\
Reference & ~model1~ & ~model2~ & ~model3~ & ~model4~ & ~model5~ \\
\midrule $L^*$ & 657 & 2196 & 2196 & 5340 & 5340 \\
KV & 256 & 1671 & 1672 & 2128 & 2128 \\
\lsharp{} & 212 & 862 & 862 & 1730 & 1730 \\
\midrule \pdlstar~ & 657 & 2325 & 2650 & 2997 & 4520 \\
\IKV & 160 & {\textcolor{teal}{814}} & {\textcolor{teal}{458}} & {\textcolor{teal}{770}} & {\textcolor{teal}{841}} \\
\midrule \adaptivelsharp & {\textcolor{teal}{146}} & 918 & 580 & 1592 & 1043 \\
\midrule \lsharprebuild{} & 161 & 954 & 749 & 1765 & 1382 \\
\lsharpmatch & 167 & 956 & 590 & 1775 & 1178 \\
\lsharpapproxmatch & 152 & 920 & 590 & 1602 & 1178 \\
\lsharprebuildmatch & 161 & 954 & 580 & 1765 & 1043 \\
\bottomrule
\end{tabular}}
\label{tab:my_exp2_philips}
\end{table}

%% file: finalversionfigures/exp2_tableOpenSSL.tex
\renewcommand{\arraystretch}{1.4}
\begin{table}[t]
\centering
\caption{Mean inputs for learning an OpenSSL model with a reference.}
\resizebox{.99\textwidth}{!}{
\begin{tabular}{lrrrrrrrrrrrrrrrrrrr}
\toprule
Algorithm & 097c & 097e & 098f & 098l & 098m & 098s & 098u & 098za & 100 & 100f & 100h & 100m & 101 & 101h & 101k & 102 & 110pre1 \\
Reference & 097 & 097c & 097e & 098f & 098l & 098m & 098s & 098u & 098m & 100 & 100f & 100h & 100h & 100 & 101h & 101k & 102 \\
\midrule $L^*$ & 21273 & 21273 & 31608 & 2408 & 2065 & 2065 & 2506 & 1820 & 2065 & 2065 & 2506 & 1820 & 3143 & 1820 & 1477 & 1281 & 1134 \\
KV & 22434 & 19376 & 24754 & 6103 & 4267 & 4504 & 5634 & 3686 & 4267 & 4659 & 6492 & 3545 & 8423 & 3545 & 2764 & 2239 & 1178 \\
\lsharp{} & 23512 & 23305 & 25412 & 6786 & 5145 & 5934 & 6562 & 3684 & 5145 & 5819 & 6283 & 3983 & 8938 & 3983 & 3075 & 2293 & 1452 \\
\midrule \pdlstar~ & 5155 & 5155 & 5155 & 3203 & 2065 & 2065 & 2317 & 2363 & 2065 & 2065 & 2317 & 2430 & 3820 & 2430 & 2115 & 1281 & 1134 \\
\IKV & {\textcolor{teal}{3290}} & 4872 & 1506 & 2945 & 2326 & 2875 & 3789 & 792 & 876 & 3153 & 3398 & 792 & 2033 & 831 & 636 & 977 & {\textcolor{teal}{376}} \\
\midrule \adaptivelsharp & 3808 & {\textcolor{teal}{1391}} & {\textcolor{teal}{1391}} & 861 & {\textcolor{teal}{756}} & {\textcolor{teal}{737}} & {\textcolor{teal}{2100}} & {\textcolor{teal}{638}} & {\textcolor{teal}{751}} & {\textcolor{teal}{737}} & {\textcolor{teal}{1953}} & {\textcolor{teal}{638}} & {\textcolor{teal}{1778}} & {\textcolor{teal}{638}} & {\textcolor{teal}{514}} & {\textcolor{teal}{461}} & 665 \\
\midrule \lsharprebuild{} & 19632 & 1397 & 1397 & {\textcolor{teal}{843}} & {\textcolor{teal}{756}} & {\textcolor{teal}{737}} & 2109 & 642 & {\textcolor{teal}{751}} & {\textcolor{teal}{737}} & 1963 & 642 & 1791 & 642 & 518 & 1545 & 1538 \\
\lsharpmatch & 23346 & 17503 & 1417 & 7358 & 5458 & 6437 & 6441 & 5081 & 768 & 6619 & 7054 & 4804 & 1800 & 4804 & 532 & 2606 & 601 \\
\lsharpapproxmatch & 3558 & 17201 & 1417 & 2764 & 3856 & 2641 & 3365 & 657 & 768 & 3110 & 3974 & 657 & 1800 & 664 & 532 & 464 & 667 \\
\lsharprebuildmatch & 19683 & 1397 & {\textcolor{teal}{1391}} & {\textcolor{teal}{843}} & {\textcolor{teal}{756}} & {\textcolor{teal}{737}} & 2109 & {\textcolor{teal}{638}} & {\textcolor{teal}{751}} & {\textcolor{teal}{737}} & 1963 & {\textcolor{teal}{638}} & 1782 & {\textcolor{teal}{638}} & 520 & 937 & 663 \\
\bottomrule
\end{tabular}
}
\label{tab:my_exp2_openssl}
\end{table}

%% file: appendixexample.tex
\newpage
\section{Detailed Example Run of \adaptivelsharp} \label{app:F}
In this section, we give a detailed explanation of how $\S$ can be learned with references $\R_1$ and $\R_3$ using \adaptivelsharp. From the references, we derive the following state cover and separating family:
\[ P = P^{\R_1} \cup P^{\R_3} = \{ \varepsilon, c, ca \} \cup \{ \varepsilon, b, bb, bbb \} = \{\varepsilon, c, ca, b, bb, bbb \}\]
\[ W_{r_0}=W_{r_1}=\{c,ac\}, W_{r_2}=\{c\},W_{s_0}=W_{s_1}=W_{s_2}=W_{s_3}=\{c,b,bb\} \]

\begin{enumerate}
    \item \adaptivelsharp always start with an observation tree containing only the root node.
    \item The first rule we apply is the \textbf{rebuilding} rule. We apply this rule with $q = q' = q_0$ (the root state) and $i = c$ because the conditions hold:
    \begin{itemize}
        \item $\delta^{\Obs}(q_0,c) \notin \basis$ because 
        $\delta^{\Obs}(q_0,c){\uparrow}$,
        \item $\neg(q' \apart \delta^{\Obs}(q,i)) = \neg(q_0 \apart \delta^{\Obs}(q_0,c))$ because $\delta^{\Obs}(q_0,c){\uparrow}$,
        \item $\access^{\Obs}(q,i)=\access^{\Obs}(q_0,c)=c \in P$,
        \item $\access^{\Obs}(q')=\access^{\Obs}(q_0)=\varepsilon \in P$,
        \item $\delta^{\Obs}(q_0,cac){\uparrow} \land \delta^{\Obs}(q_0,ac){\uparrow}$ with $\mathsf{sep}(\delta^{\R}(\access^{\Obs}(q)i),\delta^{\R}(\access^{\Obs}(q'))) = \mathsf{sep}(\delta^{\R}(c),\delta^{\R}(\varepsilon)) = ac = \sigma$.
    \end{itemize}
    We execute $OQ(cac)$ and $OQ(ac)$.
    \item We can now apply \textbf{prioritized promotion} with $q_1$ because $ac \vdash q_0 \apart q_1$. The resulting observation tree looks as follows:
    \begin{figure}[H]
        \centering
            \resizebox{.25\textwidth}{!}{
                \includegraphics{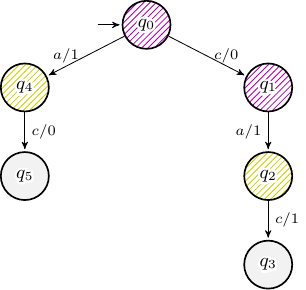}
            }
        \end{figure}
    \item Next, we try to promote the state reached by $ca$ in $\R_1$. Note that we cannot apply rebuilding with $q = q_1$, $q' = q_0$ and $i = a$ because $\mathsf{sep}(\delta^{\R}(ca),\delta^{\R}(\varepsilon))=c$ and $\delta^{\Obs}(c){\converges}$ and $\delta^{\Obs}(cac){\converges}$.\\
    We do apply \textbf{rebuilding} with $q = q_1$, $q' = q_1$ and $i = a$. All the conditions hold and $\sigma = c$. This leads to output queries $OQ(cac)$ and $OQ(cc)$.
    \item We can apply \textbf{prioritized promotion} with $q_2$ because $c \vdash q_1 \apart q_2$ and $c \vdash q_0 \apart q_2$.
    \item We again use the rebuilding rule for $q = q_0, q' = q_0$ and $i = b$. All the conditions hold and we use $\sigma = \mathsf{sep}(\delta^{\R}(\access^{\Obs}(q)i),\delta^{\R}(\access^{\Obs}(q'))) = \mathsf{sep}(\delta^{\R}(b),\delta^{\R}(\varepsilon)) = bb$. We execute the queries $OQ(bbb)$ and $OQ(bb)$. This leads to the following observation tree. Note that we cannot promote $q_7$.
    \begin{figure}[H]
        \centering
            \resizebox{.4\textwidth}{!}{
            \includegraphics{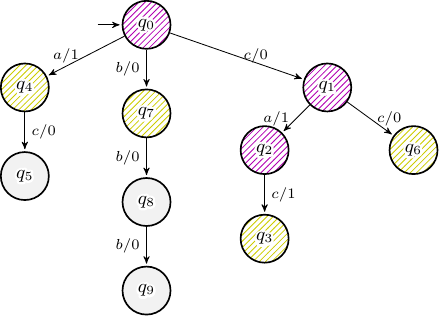}}
        \end{figure}
    As stated in the Section~\ref{sec:multiple}, we only apply the rebuilding rule with $\access^{\Obs}(q,i)$ and $\access^{\Obs}(q')$ from the same reference model. We have not explored $bb, bbb \in P$ yet but because these access sequences do not reach frontier states, we cannot apply the rebuilding rule any further. The prioritized promotion rule can also not be applied. Therefore, we move on to the other \adaptivelsharp rules.
    \item Next, we apply the \textbf{extension} rule for all the current basis states. This means the following output queries are performed: $OQ(cb)$, $OQ(caa)$, $OQ(cab)$. This results in the following observation tree.
    \begin{figure}[H]
        \centering
            \resizebox{.5\textwidth}{!}{
            \includegraphics{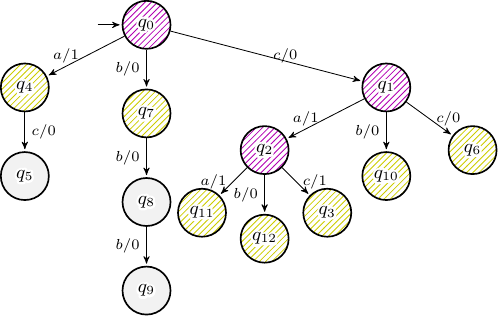}}
        \end{figure}
    \item We compute the matching table and then apply \textbf{prioritized separation}:
    \begin{table}[H]
        \centering
        \begin{tabular}{p{0.8cm}p{1cm}p{0.8cm}|p{1cm}p{1cm}p{1cm}p{1cm}p{1cm}p{1cm}p{1cm}}
            state & match & $\mathsf{mdeg}$ & $r_0$ & $r_1$ & $r_2$ & $s_0$ & $s_1$ & $s_2$ & $s_3$ \\ \hline 
            $q_0$ & $r_0$ & 1.0 & 7/7 & 5/7 & 4/7 & 0/4 & 0/4 & 0/4 & 1/4\\
            $q_1$ & $r_1$ & 1.0 & 3/4 & 4/4 & 2/4 & 0/4 & 0/4 & 0/4 & 1/4\\
            $q_2$ & $r_2$ & 1.0 & 1/2 & 1/2 & 2/2 & 0/2 & 0/2 & 0/2 & 1/2\\
        \end{tabular}
    \end{table}
    \begin{itemize}
        \item To separate $q_4$ from the basis, we can use the separating sequence $ac$ because $ac \in W_{r_1}$ and $r_1$ is the expected matching reference state of $q_4$ because $\delta^{\Obs}(q_0,a)=q_4$, $q_0 \approxmatches r_0$ and $\delta^{\R}(r_0,a)=r_1$. Therefore, we execute $OQ(aac)$.
        \item To separate $q_6$ from the basis, we can use the separating sequence $ac$ because $ac \in W_{r_0}$ and $r_0$ is the expected matching reference state of $q_6$ because $\delta^{\Obs}(q_1,c)=q_6$, $q_1 \approxmatches r_1$ and $\delta^{\R}(r_1,c)=r_0$. Therefore, we execute $OQ(ccac)$.
        \item To separate $q_{11}$ from the basis, we can use the separating sequence $ac$ because $ac \in W_{r_1}$ and $r_1$ is the expected matching reference state of $q_{11}$ because $\delta^{\Obs}(q_2,a)=q_{11}$, $q_2 \approxmatches r_2$ and $\delta^{\R}(r_2,c)=r_1$. Therefore, we execute $OQ(caaac)$.
        \item To separate $q_3$ from the basis, we can use the separating sequence $c$ because $c \in W_{r_2}$ and $r_2$ is the expected matching reference state of $q_3$ because $\delta^{\Obs}(q_2,c)=q_{11}$, $q_2 \approxmatches r_2$ and $\delta^{\R}(r_2,c)=r_2$. Therefore, we execute $OQ(cacc)$.
    \end{itemize}
    \item Next, we apply the \textbf{promotion} rule for $q_4$ because $a \vdash q_4 \apart q_0$, $a \vdash q_4 \apart q_1$ and $a \vdash q_4 \apart q_2$. The resulting observation tree looks as follows:
    \begin{figure}[H]
        \centering
            \resizebox{.75\textwidth}{!}{
            \includegraphics{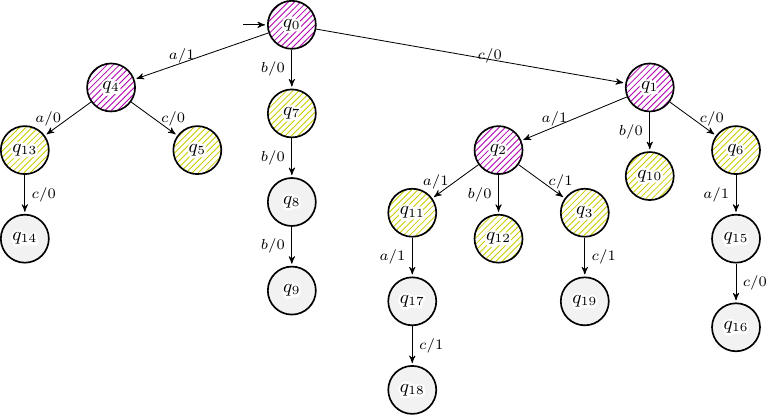}}
        \end{figure}
    \item Next, we apply \textbf{extension} with $q_4$ and input $b$, resulting in $OQ(ab)$.
    \item We perform another round of \textbf{prioritized separation} with the following matching table:
    \begin{table}[H]
        \centering
        \begin{tabular}{p{1cm}p{1cm}p{1cm}|p{1cm}p{1cm}p{1cm}p{1cm}p{1cm}p{1cm}p{1cm}}
            state & match & $\mathsf{mdeg}$ & $r_0$ & $r_1$ & $r_2$ & $s_0$ & $s_1$ & $s_2$ & $s_3$\\ \hline 
            $q_0$ & $r_0$ & 0.857 & 12/14 & 8/14 & 7/14 &  2/6 &  2/6 &  1/6 &  3/6\\
            $q_1$ & $r_1$ & 1.0 & 5/9 &  9/9 &  5/9 &  0/5 &  0/5 &  0/5 &  1/5\\
            $q_2$ & $r_2$ & 1.0 & 3/5 &  2/5 &  5/5 &  0/3 &  0/3 &  0/3 &  1/3\\
            $q_4$ & $s_0,s_1$ & 1.0 & 2/3 &  1/3 &  1/3  &  2/2 &  2/2 &  1/2 &  0/2
        \end{tabular}
    \end{table}
    \begin{itemize}
        \item To separate $q_6$ further, we execute $OQ(ccc)$.
        \item To separate $q_{11}$ further, we execute $OQ(caac)$.
        \item To separate $q_{20}$, we can use the separating sequence $cb$ because $nb \in W_{s_1} \cup W_{s_2}$ and $s_1$ and $s_2$ are the expected matching reference states of $q_{20}$ because $\delta^{\Obs}(q_4,b)=q_{20}$, $q_4 \approxmatches s_0$ and $\delta^{\R}(s_0,b)=s_1$, $q_4 \approxmatches s_1$ and $\delta^{\R}(s_1,b)=s_2$. Therefore, we execute $OQ(abb)$.
        \item To separate $q_{13}$, we can use the separating sequence $b$ because $b \in W_{s_0} \cup W_{s_1}$ and $s_0$ and $s_1$ are the expected matching reference states of $q_{13}$ because $\delta^{\Obs}(q_4,a)=q_{13}$, $q_4 \approxmatches s_0$ and $\delta^{\R}(s_0,a)=s_0$, $q_4 \approxmatches s_1$ and $\delta^{\R}(s_1,a)=s_1$. Therefore, we execute $OQ(aab)$.
    \end{itemize}
    The resulting observation tree looks as follows:
    \begin{figure}[H]
        \centering
            \resizebox{.9\textwidth}{!}{
            \includegraphics{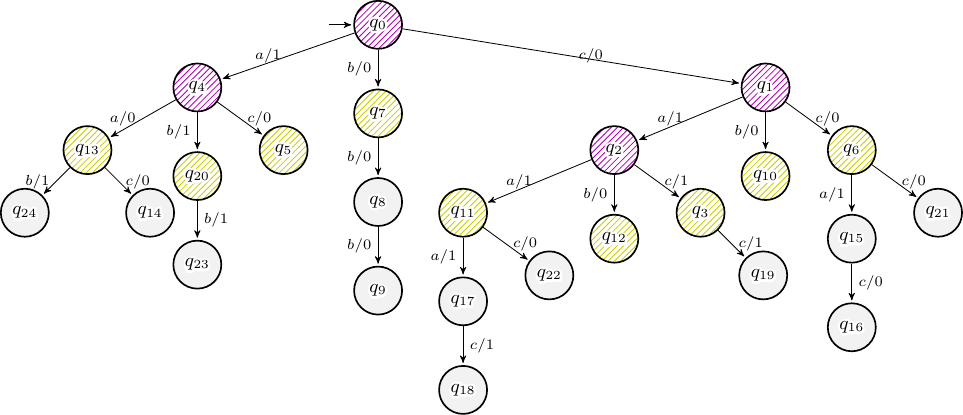}}
        \end{figure}
    \item We can no longer apply prioritized separation so we continue with standard \textbf{separation}. We use the separating sequences $b, c$ and $ac$ to separate states. Specifically, we perform the following output queries:
    \begin{itemize}
        \item $OQ(bac)$ and $OQ(bc)$ to separate $q_7$ from the basis.
        \item $OQ(cbac)$ to separate $q_{10}$ from the basis.
        \item $OQ(acac)$ to separate $q_5$ from the basis.
        \item $OQ(cabac)$ and $OQ(cabc)$ to separate $q_{12}$ from the basis.
    \end{itemize}
    The resulting observation tree looks as follows:
    \begin{figure}[H]
        \centering
            \resizebox{.99\textwidth}{!}{
            \includegraphics{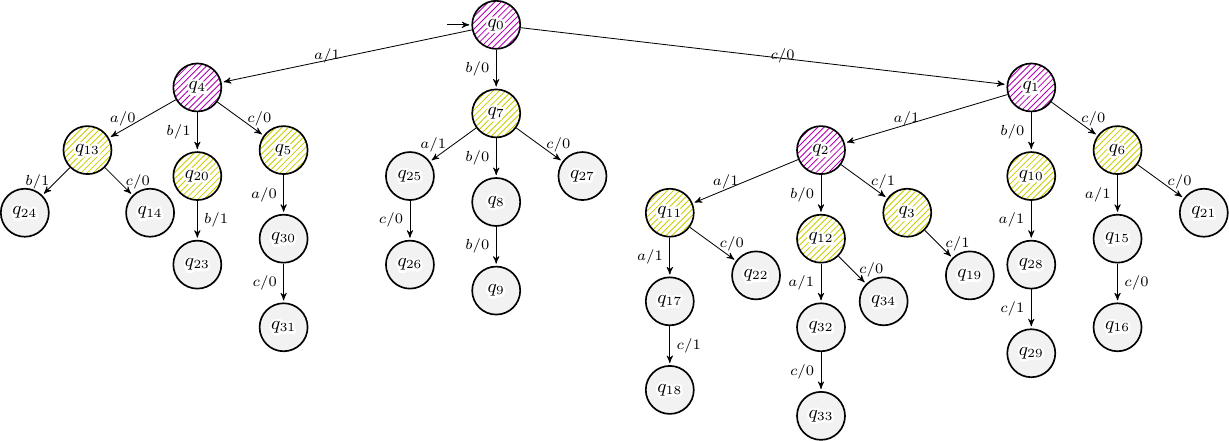}}
        \end{figure}
    \item All frontier states are identified and no frontier states are isolated. This means we can possibly perform match refinement or match separation. The matching table looks as follows:
    \begin{table}[H]
        \centering
        \begin{tabular}{p{1cm}p{1cm}p{1cm}|p{1cm}p{1cm}p{1cm}p{1cm}p{1cm}p{1cm}p{1cm}}
            state & match & $\mathsf{mdeg}$ & $r_0$ & $r_1$ & $r_2$ & $s_0$ & $s_1$ & $s_2$ & $s_3$ \\ \hline 
            $q_0$ & $r_0$ & 0.834 & 15/18 & 10/18 &  8/18 &   4/9 &   3/9 &   2/9  &   6/9\\
            $q_1$ & $r_1$ & 1.0 & 6/11 & 11/11 &  5/11  &   0/7 &   0/7 &   2/7 &   2/7\\
            $q_2$ & $r_2$ & 1.0 & 4/6 &   2/6 &   6/6  &   0/4 &   0/4 &   1/4 &   2/4\\
            $q_4$ & $s_0$ & 1.0 & 2/5 &   2/5 &   2/5  &   4/4 &   3/4 &   1/4 &   1/4
        \end{tabular}
    \end{table}
    Since every basis state is matched with exactly one reference state, match refinement is not applicable. However, we can apply \textbf{match separation} with $q = q_4$, $q' = q_4$, $p=s_0$ and $i = b$ because $\delta^{\Obs}(q_4,b)=q_{20} \in \frontier$, $\neg(q_{20} \apart q_4)$, $\delta^{\R}(s_0,b)=s_1$, $s_0 \approxmatches q_4$, and for all $q \in \basis$, $s_1$ is not the match.
    We use $\sigma = bb$ because $bb \vdash s_1 \apart q_4$ and execute $OQ(abbb)$. This leads to the following observation tree:
    \begin{figure}[H]
        \centering
            \resizebox{.99\textwidth}{!}{
            \includegraphics{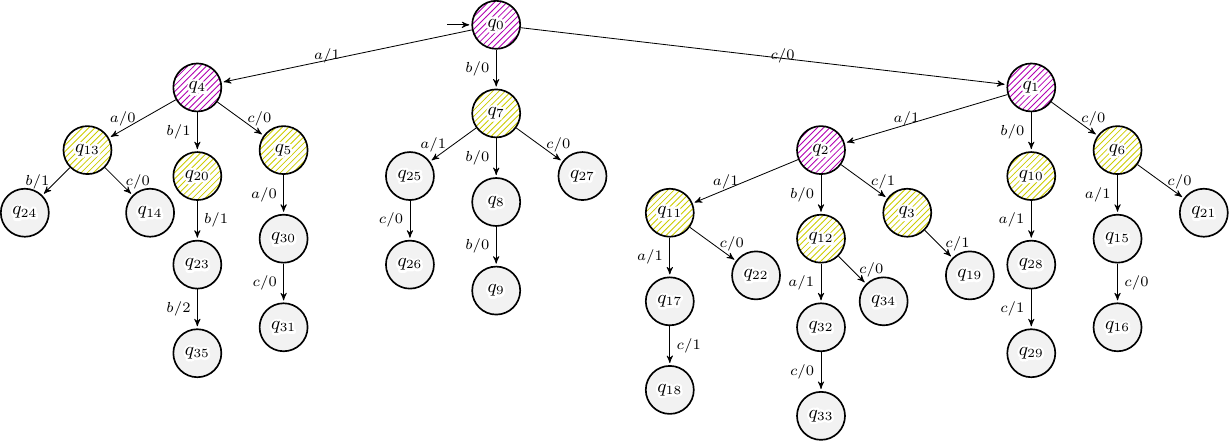}}
        \end{figure}
    \item The match separation led to isolation of $q_{20}$ which we can now add to the basis with the \textbf{promotion} rule.
    \item Additionally, we can apply \textbf{promotion} for state $q_{23}$ because it is the only state with output 2 for input $b$.
    Notice that we have found all the states in $\S$ after promoting $q_{23}$. 
    \item Next, we apply the extension rule several times 
    \begin{itemize}
        \item $OQ(aba)$ and $OQ(abc)$ for $q_{20}$,
        \item $OQ(abba)$ and $OQ(abbc)$ for $q_{23}$.
    \end{itemize}
    \item We apply the \textbf{prioritized separation} rule several times
    \begin{itemize}
        \item $OQ(abac)$ and $OQ(ababb)$ for $q_{36}$,
        \item $OQ(aabb)$ for $q_{13}$,
        \item $OQ(abbbc)$ and $OQ(abbbb)$ for $q_{34}$,
        \item $OQ(abbac)$ and $OQ(abbab)$ for $q_{38}$,
    \end{itemize}
    \item Next, we apply \textbf{separation} a few more times.
    \begin{itemize}
        \item $OQ(abcac)$ and $OQ(abcabb)$ for $q_{37}$,
        \item $OQ(abbcac)$ and $OQ(abbcabb)$ for $q_{39}$,
        \item $OQ(abbbac)$ for $q_{35}$,
        \item $OQ(acabb)$ for $q_{5}$,
    \end{itemize}
    This leads to the following final observation tree and matching table:
    \begin{figure}[H]
        \centering
            \resizebox{.99\textwidth}{!}{
            \includegraphics{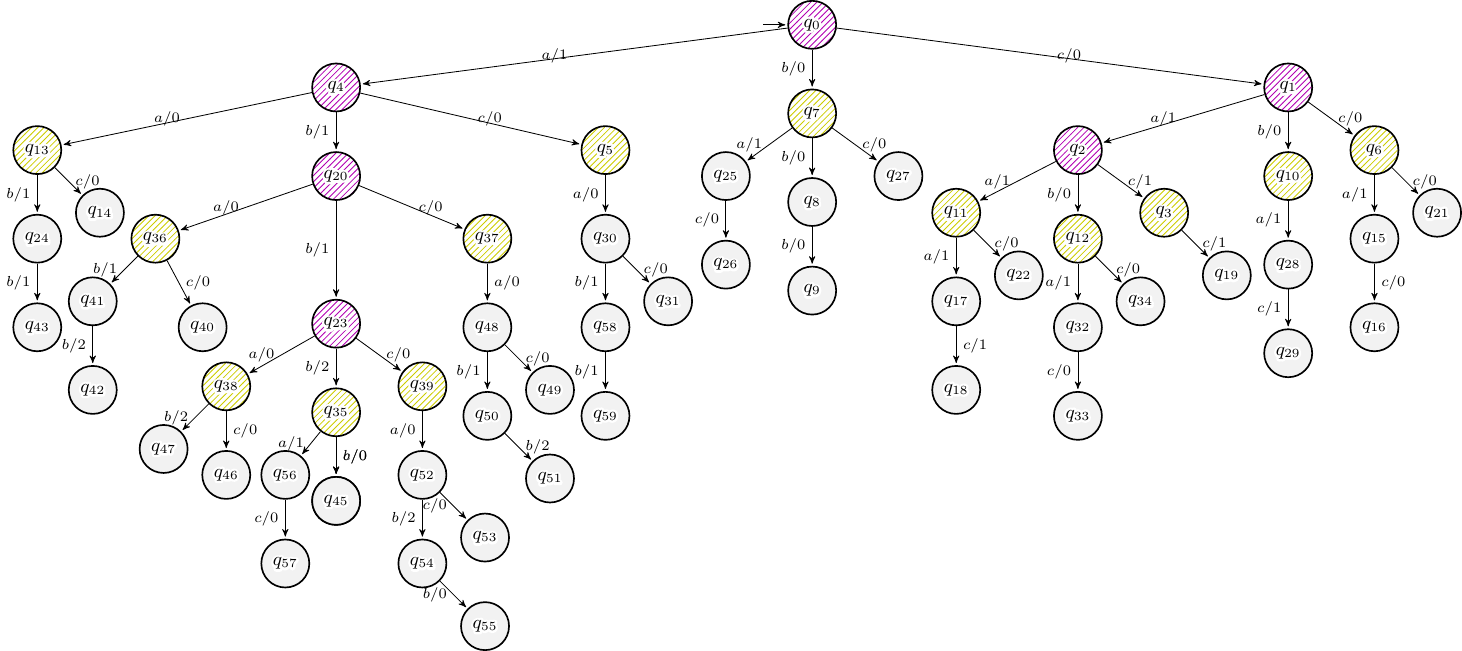}
            }
        \end{figure}
    \begin{table}[H]
        \centering
        \begin{tabular}{p{1cm}p{1cm}p{1cm}|p{1cm}p{1cm}p{1cm}p{1cm}p{1cm}p{1cm}p{1cm}}
            state & match & $\mathsf{mdeg}$ & $r_0$ & $r_1$ & $r_2$ & $s_0$ & $s_1$ & $s_2$ & $s_3$\\ \hline 
            $q_0$ & $r_0$ & 0.834 & 15/18 & 10/18 &  8/18  & 12/18 &  7/18 &  5/18 & 14/18\\
            $q_1$ & $r_1$ & 1.0 & 6/11 & 11/11 &  5/11  &   0/7 &   0/7 &   2/7 &   2/7\\
            $q_2$ & $r_2$ & 1.0 & 4/6 &   2/6 &   6/6  &   0/4 &   0/4 &   1/4 &   2/4\\
            $q_4$ & $s_0$ & 0.924 & 2/5 & 2/5 & 2/5  & 12/13 & 7/13 &  4/13 & 5/13\\
            $q_{20}$ & $s_1$ & 0.889 & 2/5 &   2/5 &   2/5  &   4/9 &   8/9 &   4/9 &   3/9\\
            $q_{23}$ & $s_2$ & 0.8 & 2/5 &   2/5 &   2/5  &   1/5 &   1/5 &   4/5 &   2/5\\
        \end{tabular}
    \end{table}
    \item Next, the \textbf{equivalence} rule is used to construct a hypothesis from the observation tree. This hypothesis is correct so the algorithm terminates.
\end{enumerate}

%% file: main.bbl
\begin{thebibliography}{10}

\bibitem{DBLP:journals/corr/abs-2209-14031}
Bernhard~K. Aichernig, Edi Muskardin, and Andrea Pferscher.
\newblock Active vs. passive: {A} comparison of automata learning paradigms for
  network protocols.
\newblock In {\em FMAS/ASYDE@SEFM}, volume 371 of {\em {EPTCS}}, pages 1--19,
  2022.

\bibitem{DBLP:journals/iandc/Angluin87}
Dana Angluin.
\newblock Learning regular sets from queries and counterexamples.
\newblock {\em Inf. Comput.}, 75(2):87--106, 1987.

\bibitem{DBLP:journals/sosym/AslamCSB20}
Kousar Aslam, Loek Cleophas, Ramon R.~H. Schiffelers, and Mark van~den Brand.
\newblock Interface protocol inference to aid understanding legacy software
  components.
\newblock {\em Softw. Syst. Model.}, 19(6):1519--1540, 2020.

\bibitem{DBLP:conf/birthday/BainczykSH20}
Alexander Bainczyk, Bernhard Steffen, and Falk Howar.
\newblock Lifelong learning of reactive systems in practice.
\newblock In {\em The Logic of Software. {A} Tasting Menu of Formal Methods},
  volume 13360 of {\em {LNCS}}, pages 38--53. Springer, 2022.

\bibitem{DBLP:journals/fmsd/ChakiCSS08}
Sagar Chaki, Edmund~M. Clarke, Natasha Sharygina, and Nishant Sinha.
\newblock Verification of evolving software via component substitutability
  analysis.
\newblock {\em Formal Methods Syst. Des.}, 32(3):235--266, 2008.

\bibitem{DBLP:conf/ifm/DamascenoMS19}
Carlos Diego~Nascimento Damasceno, Mohammad~Reza Mousavi, and Adenilso
  da~Silva~Sim{\~{a}}o.
\newblock Learning to reuse: Adaptive model learning for evolving systems.
\newblock In {\em {IFM}}, volume 11918 of {\em {LNCS}}, pages 138--156.
  Springer, 2019.

\bibitem{DBLP:conf/nordsec/Ruiter16}
Joeri de~Ruiter.
\newblock A tale of the {OpenSSL} state machine: {A} large-scale black-box
  analysis.
\newblock In {\em NordSec}, volume 10014 of {\em {LNCS}}, pages 169--184, 2016.

\bibitem{DBLP:conf/uss/RuiterP15}
Joeri de~Ruiter and Erik Poll.
\newblock Protocol state fuzzing of {TLS} implementations.
\newblock In {\em {USENIX} Security Symposium}, pages 193--206. {USENIX}
  Association, 2015.

\bibitem{DBLP:conf/birthday/FerreiraHS22}
Tiago Ferreira, Gerco van Heerdt, and Alexandra Silva.
\newblock Tree-based adaptive model learning.
\newblock In {\em A Journey from Process Algebra via Timed Automata to Model
  Learning}, volume 13560 of {\em {LNCS}}, pages 164--179. Springer, 2022.

\bibitem{DBLP:journals/igpl/GrocePY06}
Alex Groce, Doron~A. Peled, and Mihalis Yannakakis.
\newblock Adaptive model checking.
\newblock {\em Log. J. {IGPL}}, 14(5):729--744, 2006.

\bibitem{DBLP:conf/dagstuhl/HowarS16}
Falk Howar and Bernhard Steffen.
\newblock Active automata learning in practice - an annotated bibliography of
  the years 2011 to 2016.
\newblock In {\em Machine Learning for Dynamic Software Analysis}, volume 11026
  of {\em {LNCS}}, pages 123--148. Springer, 2018.

\bibitem{DBLP:conf/fmics/HuistraMP18}
David Huistra, Jeroen Meijer, and Jaco van~de Pol.
\newblock Adaptive learning for learn-based regression testing.
\newblock In {\em {FMICS}}, volume 11119 of {\em {LNCS}}, pages 162--177.
  Springer, 2018.

\bibitem{DBLP:conf/rv/IsbernerHS14}
Malte Isberner, Falk Howar, and Bernhard Steffen.
\newblock The {TTT} algorithm: {A} redundancy-free approach to active automata
  learning.
\newblock In {\em {RV}}, volume 8734 of {\em {LNCS}}, pages 307--322. Springer,
  2014.

\bibitem{DBLP:conf/cav/IsbernerHS15}
Malte Isberner, Falk Howar, and Bernhard Steffen.
\newblock The open-source learnlib - {A} framework for active automata
  learning.
\newblock In {\em {CAV} {(1)}}, volume 9206 of {\em {LNCS}}, pages 487--495.
  Springer, 2015.

\bibitem{DBLP:books/daglib/0041035}
Michael~J. Kearns and Umesh~V. Vazirani.
\newblock {\em An Introduction to Computational Learning Theory}.
\newblock {MIT} Press, 1994.
\newblock URL:
  \url{https://mitpress.mit.edu/books/introduction-computational-learning-theory}.

\bibitem{KrugerZenodo}
Loes Kruger, Sebastian Junges, and Jurriaan Rot.
\newblock {State Matching and Multiple References in Adaptive Active Automata
  Learning: Supplementary Material}, June 2024.
\newblock \href {https://doi.org/10.5281/zenodo.12517574}
  {\path{doi:10.5281/zenodo.12517574}}.

\bibitem{DBLP:conf/atva/MuskardinAPPT21}
Edi Muskardin, Bernhard~K. Aichernig, Ingo Pill, Andrea Pferscher, and Martin
  Tappler.
\newblock {AALpy}: An active automata learning library.
\newblock In {\em {ATVA}}, volume 12971 of {\em {LNCS}}, pages 67--73.
  Springer, 2021.

\bibitem{DBLP:conf/birthday/NeiderSVK97}
Daniel Neider, Rick Smetsers, Frits~W. Vaandrager, and Harco Kuppens.
\newblock Benchmarks for automata learning and conformance testing.
\newblock In {\em Models, Mindsets, Meta}, volume 11200 of {\em {LNCS}}, pages
  390--416. Springer, 2018.

\bibitem{DBLP:conf/ifm/SchutsHV16}
Mathijs Schuts, Jozef Hooman, and Frits~W. Vaandrager.
\newblock Refactoring of legacy software using model learning and equivalence
  checking: An industrial experience report.
\newblock In {\em {IFM}}, volume 9681 of {\em {LNCS}}, pages 311--325.
  Springer, 2016.

\bibitem{DBLP:conf/lata/SmetsersMJ16}
Rick Smetsers, Joshua Moerman, and David~N. Jansen.
\newblock Minimal separating sequences for all pairs of states.
\newblock In {\em {LATA}}, volume 9618 of {\em {LNCS}}, pages 181--193.
  Springer, 2016.

\bibitem{DBLP:journals/corr/abs-1904-07075}
Martin Tappler, Bernhard~K. Aichernig, and Roderick Bloem.
\newblock Model-based testing {IoT} communication via active automata learning.
\newblock {\em CoRR}, abs/1904.07075, 2019.

\bibitem{DBLP:journals/cacm/Vaandrager17}
Frits~W. Vaandrager.
\newblock Model learning.
\newblock {\em Commun. {ACM}}, 60(2):86--95, 2017.

\bibitem{DBLP:conf/tacas/VaandragerGRW22}
Frits~W. Vaandrager, Bharat Garhewal, Jurriaan Rot, and Thorsten Wi{\ss}mann.
\newblock A new approach for active automata learning based on apartness.
\newblock In {\em {TACAS} {(1)}}, volume 13243 of {\em {LNCS}}, pages 223--243.
  Springer, 2022.

\bibitem{DBLP:conf/cbse/WindmullerNSHB13}
Stephan Windm{\"{u}}ller, Johannes Neubauer, Bernhard Steffen, Falk Howar, and
  Oliver Bauer.
\newblock Active continuous quality control.
\newblock In {\em {CBSE}}, pages 111--120. {ACM}, 2013.

\bibitem{DBLP:conf/wcre/YangASLHCS19}
Nan Yang, Kousar Aslam, Ramon R.~H. Schiffelers, Leonard Lensink, Dennis
  Hendriks, Loek Cleophas, and Alexander Serebrenik.
\newblock Improving model inference in industry by combining active and passive
  learning.
\newblock In {\em {SANER}}, pages 253--263. {IEEE}, 2019.

\end{thebibliography}
